\documentclass[letterpaper]{article}
\usepackage{times}

\usepackage{mathrsfs}
\usepackage{graphicx}
\usepackage{amsmath}
\usepackage{amssymb}
\usepackage{amsthm}
\usepackage{algorithm}
\usepackage{algpseudocode}
\usepackage{multirow} 
\usepackage{footmisc}

\usepackage{latexsym}
\usepackage{amsfonts}
\usepackage{stmaryrd}
\usepackage{hyperref}

\newcommand{\veca}{{{\mathbf{a}}}}
\newcommand{\vecb}{{{\mathbf{b}}}}
\newcommand{\vecu}{{{\mathbf{u}}}}
\newcommand{\vecp}{{{\mathbf{p}}}}
\newcommand{\calL}{{{\mathcal{L}}}}
\newcommand{\calT}{{{\mathcal{T}}}}
\newcommand{\calS}{{{\mathcal{S}}}}

\newcommand{\calE}{{{\mathcal{E}}}}
\newcommand{\calA}{{{\mathcal{A}}}}
\newcommand{\score}{{{\mathrm{sc}}}}
\newcommand{\PoA}{{{\mathrm{PoA}}}}
\newcommand{\eps}{\varepsilon}

\newcommand{\pref}{\succ}

\newtheorem{proposition}{{\bf Proposition}}
\newtheorem{definition}{{\bf Definition}}
\newtheorem{theorem}{{\bf Theorem}}
\newtheorem{lemma}{{\bf Lemma}}

\newtheorem{corollary}{{\bf Corollary}}

\newtheorem{remark}{\bf Remark}

\newtheorem{example}{{\bf Example}}

\begin{document}


\title{Equilibria of Plurality Voting: Lazy and Truth-biased Voters}

\author{Edith Elkind\thanks{University of Oxford} \and Evangelos Markakis\thanks{Athens University of Economics and Business} \and Svetlana Obraztsova\thanks{National Technical University of Athens} \and Piotr Skowron\thanks{University of Warsaw}
}
\date{}


\maketitle
\begin{abstract}
We present a systematic study of Plurality elections with strategic voters who, in addition
to having preferences over election winners, have secondary
preferences, which govern their behavior when their vote cannot affect the election outcome. 
Specifically, we study two models that have been recently considered in the
literature~\cite{des-elk:c:eq,obr-mar-tho:c:truth-biased}:
{\em lazy} voters, who prefer to abstain when they are not pivotal,
and {\em truth-biased} voters, who prefer to vote truthfully when they are not pivotal.
We extend prior work by investigating the behavior of both lazy and truth-biased
voters under different tie-breaking rules (lexicographic rule, random voter rule, random candidate rule).
Two of these six combinations of secondary preferences and a tie-breaking rule
have been studied in prior work. In order to understand the impact of different 
secondary preferences and tie-breaking rules on the election outcomes,
we study the remaining four combinations. 
We characterize pure Nash equilibria (PNE) of the resulting 
strategic games and study the complexity of related computational problems.
Our results extend to settings where some of the voters may be non-strategic.
\end{abstract}

\maketitle

\section{Introduction}

Plurality voting is a popular tool for collective decision-making 
in many domains, including both human societies and multiagent systems.
Under this voting rule, each voter is supposed to vote 
for her most favorite candidate (or abstain); the winner is then the candidate that receives the highest
number of votes. If several candidates 
have the highest score, the winner is chosen among them using a {\em tie-breaking rule};
popular tie-breaking rules include the {\em lexicographic rule}, which
imposes a fixed priority order over the candidates; 
the {\em random candidate rule}, which picks one of the tied candidates
uniformly at random; and the {\em random voter rule}, 
which picks the winner among the tied candidates according 
to the preferences of a randomly chosen voter.  

In practice, voters are often {\em strategic}, i.e., they may vote
non-truthfully if they can benefit from doing so. In that case, an election
can be viewed as a game, where the voters are the players, and each player's space of actions
includes voting for any candidate or abstaining. For deterministic rules
(such as Plurality with lexicographic tie-breaking), the behavior 
of strategic voters is determined by their preference ordering, i.e.,
a ranking of the candidates, whereas for randomized rules
a common approach is to specify utility functions for the voters;
i.e., the voters are assumed to maximize their {\em expected utility}
under the lottery induced by tie-breaking.
The outcome of the election can then be identified with a pure Nash equilibrium (PNE)
of the resulting game.

However, for the Plurality voting game with $3$ or more voters, this approach fails to provide
a useful prediction of voting behavior: for each candidate $c$
there is a PNE where $c$ is the unique winner, irrespective 
of the voters' preferences. Indeed, if there are at least $3$ voters,
the situation where all of them vote for $c$ is a PNE, as
no voter can unilaterally change the election outcome.
However, such equilibria may disappear if we use a more refined model
of voters' preferences that captures additional aspects 
of their decision-making. For instance, in practice, 
if a voter feels that her vote is unlikely to have any effect on the election
outcome, she may decide to abstain from the election.
Also, voters may be averse to lying about their preferences, 
in which case they can be expected to vote for their top candidate unless
there is a clear strategic reason to vote for someone else.
By taking into account these aspects of voters' preferences,
we obtain a more faithful model of their behavior.

The problem of characterizing and computing the equilibria of Plurality voting, 
both for ``lazy'' voters (i.e., ones who prefer to abstain when they are not pivotal)
and for ``truth-biased'' voters (ones who prefer to vote truthfully when they are not pivotal),
has recently received a considerable amount of attention.
However, it is difficult to compare the existing results,
since they rely on different tie-breaking rules. In particular, 
\cite{des-elk:c:eq}, who study lazy voters, 
use the random candidate tie-breaking rule, and~\cite{obr-mar-tho:c:truth-biased} consider truth-biased
voters and the lexicographic tie-breaking rule. Thus, it is not clear whether
the differences between the results in these papers
can be attributed to voters' secondary preferences or to the tie-breaking rule.

The primary goal of our paper is to tease out the effects of different features
of these models, by systematically considering various combinations
of secondary preferences and tie-breaking rules. We consider two types
of secondary preferences (lazy voters and truth-biased voters) and three
tie-breaking rules (the lexicographic rule, the random voter rule, and the random candidate rule);
while two of these combinations have been studied earlier by Desmedt and Elkind~\cite{des-elk:c:eq}
and Obraztsova et al.~\cite{obr-mar-tho:c:truth-biased}, to the best of our knowledge, the 
remaining four possibilities have not been considered before.  
For each of the new scenarios, we characterize the set of PNE for the resulting game;
in doing so, we also fill in a gap in the characterization of
Desmedt and Elkind for lazy voters
and random candidate tie-breaking. We then consider the problems
of deciding whether a given game admits a PNE and whether a given
candidate can be a co-winner/unique winner in some PNE of a given game.
For all settings we consider, we determine the computational
complexity of each of these problems, classifying them as either polynomial-time
solvable or NP-complete. We use our characterization results to
analyze the impact of various features of our models on the election outcomes.
Finally, we extend our results to the setting where some
of the voters may be {\em principled}, i.e., are guaranteed to vote truthfully.


\smallskip

\noindent{\bf Related Work\quad}

Equilibria of Plurality voting have been investigated by a number
of researchers, starting with~\cite{far:b:voting}.
However, most of the earlier works either consider 
solution concepts other than pure Nash equilibria, such 
as iterative elimination of dominated strategies~\cite{mou:j:dominance,dhi-loc:j:dominance}, 
or assume that voters have incomplete information
about each others' preferences~\cite{mye-web:j:voting}.
Both types of secondary preferences
(lazy voters and truth-biased voters) appear in the social choice literature, 
see, respectively, \cite{bat:j:abstentions,bor:j:costly,sin-ian:j:costly}
and \cite{dut-sen:j:nash,lom-yos:j:nash}. 
In computational social choice,
truth-biased voters have been considered by Meir et al.~\cite{mei-pol:c:convergence} 
in the context of dynamics of Plurality voting; subsequently,
Plurality elections with truth-biased voters have been investigated
empirically by Thompson et al.~\cite{tho-lev-ley:c:empirical} and theoretically 
by Obraztsova et al.~\cite{obr-mar-tho:c:truth-biased}. To the best of our knowledge, 
the only paper to study computational aspects of Plurality voting with
lazy voters is that of Desmedt and Elkind~\cite{des-elk:c:eq}. 

Our approach to tie-breaking is well-grounded in existing works.
Lexicographic tie-breaking is standard in the computational 
social choice literature.
The random candidate rule has been discussed by Desmedt and Elkind~\cite{des-elk:c:eq}, 
and, more recently, by Obraztsova, Elkind and Hazon~\cite{obr-elk-haz:c:ties}
and Obraztsova and Elkind~\cite{obr-elk:c:ties2}.
The random voter rule is used to break ties under the Schulze method~\cite{sch:j:schulze};
complexity of manipulation under this tie-breaking rule
has been studied by Aziz et al.~\cite{azi-gas-mat:c:random-voter}.  

\section{Preliminaries}
\noindent
For any positive integer $t$, we denote the set $\{1, \dots, t\}$ by $[t]$.
We consider elections with a set of {\em voters} $N=[n]$ and 
a set of {\em alternatives}, or {\em candidates}, $C = \{c_1, \dots c_m\}$.
Each voter is associated with a {\em preference order}, i.e., a strict linear
order over $C$; we denote the preference order of voter $i$ by $\succ_i$.
The list $(\succ_1,\dots, \succ_n)$ is called a {\em preference profile}.
For each $i\in N$, we set $a_i$ to be the top choice of voter $i$, and 
let $\veca=(a_1, \dots, a_n)$. Given two disjoint sets of candidates $X$, $Y$
and a preference order $\succ$, we write $X \succ Y$ if in $\succ$
all candidates from $X$ are ranked above all candidates from $Y$.

We also assume that
each voter $i\in N$ is endowed with a {\em utility function} $u_i:C\to{\mathbb N}$;
$u_i(c_j)$ is the utility derived by voter $i$ if $c_j$ is the unique election winner.
We require that $u_i(c)\neq u_i(c')$ for all $i\in N$ and all $c, c'\in C$ such that $c\neq c'$.
The vector $\vecu=(u_1,\dots, u_n)$ is called the {\em utility profile}.
Voters' preference orders and utility functions are assumed to be consistent, i.e., 
for each $i\in N$ and every pair of candidates $c, c'\in C$ we have $c\succ_i c'$
if and only if $u_i(c)>u_i(c')$; when this is the case, we will also say that $\succ_i$
is {\em induced} by $u_i$. Sometimes, instead of specifying preference
orders explicitly, we will specify the utility functions only, and assume that voters'
preference orders are induced by their utility functions; on other occasions, it will be convenient
to reason in terms of preference orders.

A {\em lottery} over $C$ is a vector $\vecp = (p_1, \dots, p_m)$ with $p_j\ge 0$
for all $j\in[m]$ and $\sum_{j\in[m]} p_j=1$. 
The value $p_j$ is the probability assigned to candidate $c_j$.
The {\em expected utility} of a voter $i\in N$ from a lottery $\vecp$
is given by $\sum_{j\in[m]}u_i(c_j)p_j$. 

In this paper we consider Plurality elections. In such elections each voter $i\in N$
submits a {\em vote}, or {\em ballot}, $b_i\in C\cup\{\bot\}$; if $b_i=\bot$, voter $i$ is said to {\em abstain}. 
The list of all votes $\vecb=(b_1, \dots, b_n)$ is also called a {\em ballot vector}.  
We say that a ballot vector is {\em trivial} if $b_i=\bot$ for all $i\in N$. 
Given a ballot vector $\vecb$ and a ballot $b'$,  
we write $(\vecb_{-i}, b')$ to denote the ballot vector
obtained from $\vecb$ by replacing $b_i$ with $b'$.
The {\em score} of an alternative
$c_j$ in an election with ballot vector $\vecb$ is given by $\score(c_j, \vecb) = |\{i\in N\mid b_i = c_j\}|$.
Given a ballot vector $\vecb$, we set $M(\vecb)=\max_{c\in C}\score(c, \vecb)$ and let
$W(\vecb) = \{c\in C\mid \score(c,\vecb) = M(\vecb)\}$,
$H(\vecb) = \{c\in C\mid \score(c,\vecb) = M(\vecb)-1\}$,
$H'(\vecb) = \{c\in C\mid \score(c,\vecb) = M(\vecb)-2\}$.
The set $W(\vecb)$ is called the {\em winning set}. 
Note that if $\vecb$ is trivial then $W(\vecb)=C$.
If $|W(\vecb)|=1$ then the unique candidate
in $W(\vecb)$ is declared to be the winner. Otherwise, the winner is selected
from $W(\vecb)$ according to one of the following tie-breaking rules.
\begin{itemize}
\item [(1)]
Under the {\em lexicographic rule $R^L$}, the winner is the candidate $c_j\in W(\vecb)$
such that $j\le k$ for all $c_k\in W(\vecb)$.
\item [(2)]
Under the {\em random candidate rule $R^C$}, the winner is chosen from $W(\vecb)$ uniformly at random.
\item [(3)]
Under the {\em random voter rule $R^V$}, we select a voter from $N$ uniformly at random;
if she has voted for a candidate in $W(\vecb)$, we output this candidate, 
otherwise we ask this voter to report her most preferred candidate in $W(\vecb)$,
and output the answer. This additional elicitation step
may appear difficult to implement in practice; fortunately, we can  
show that, in equilibrium it is almost never necessary.
\end{itemize} 
Thus, the outcome of an election is a lottery over~$C$; however, for $R^L$
this lottery is degenerate, i.e., it always assigns the entire
probability mass to a single candidate.
For each $X\in\{L, C, V\}$ and each ballot vector~$\vecb$, 
let $\vecp^X(\vecb)$  denote the lottery that corresponds to applying $R^X$ to 
the set $W(\vecb)$.
Note also that for every $c_j\in C$ it holds that if $p^C_j(\vecb)\neq 0$ then $p^C_j(\vecb)\ge \frac{1}{m}$;
similarly, if $p^V_j(\vecb)\neq 0$ then $p^V_j(\vecb)\ge \frac{1}{n}$.

In what follows, we consider {\em lazy} voters, who prefer to abstain when their vote
has no effect on the election outcome,
and {\em truth-biased} voters, who never abstain, but prefer to vote truthfully when their vote
has no effect on the election outcome.
Formally, pick $\eps<\min\{\frac{1}{m}, \frac{1}{n}\}$, and consider 
a utility profile $\vecu$ and a tie-breaking rule $R^X\in\{R^C, R^V, R^L\}$.
Then 
\begin{itemize}
\item 
if voter $i$ is {\em lazy}, her utility in an election with ballot vector $\vecb$ under tie-breaking rule $R^X$ 
is given by
$$
U_i(\vecb)=
\begin{cases}
\sum_{j\in [m]}p^X_j(\vecb)u_i(c_j) & \text{if $b_i\in C$},\\
\sum_{j\in [m]}p^X_j(\vecb)u_i(c_j)+\eps & \text{if $b_i=\bot$}.
\end{cases}
$$ 
\item 
if voter $i$ is {\em truth-biased}, her utility 
in an election with ballot vector $\vecb$ under tie-breaking rule $R^X$       
is given by
$$
U_i(\vecb)=
\begin{cases}
\sum_{j\in [m]}p^X_j(\vecb)u_i(c_j) &\text{if $b_i\in C\setminus\{a_i\}$},\\
\sum_{j\in [m]}p^X_j(\vecb)u_i(c_j)+\eps &\text{if $b_i=a_i$},\\
-\infty&\text{if $b_i=\bot$}.
\end{cases}
$$
\end{itemize}
We consider settings where all voters are of the same type, i.e., 
either all voters are lazy or all voters are truth-biased;
we refer to these settings as {\em lazy} or {\em truth-biased},
respectively, and denote the former by $\calL$ and the latter by $\calT$.

In what follows, we consider all possible combinations of
settings ($\calL$, $\calT$) and tie-breaking rules
($R^L$, $R^C$, $R^V$). A combination of a setting $\calS\in\{\calL, \calT\}$,
a tie-breaking rule $R\in\{R^L, R^C, R^V\}$ and a utility
profile $\vecu$ induces a strategic game, which we will denote
by $(\calS, R, \vecu)$: in this game, the players are the voters, 
the action space of each player is $C\cup\{\bot\}$, and the players' utilities
$U_1, \dots, U_n$ for a vector of actions $\vecb$ are computed based on the setting
and the tie-breaking rule as described above. 
We say that a ballot vector $\vecb$
is a {\em pure Nash equilibrium (PNE)} of the game $(\calS, R, \vecu)$
if $U_i(\vecb)\ge U_i(\vecb_{-i},b')$ for every voter $i\in N$
and every $b'\in C\cup\{\bot\}$. 

For each setting $\calS\in\{\calL, \calT\}$ and each tie-breaking rule
$R\in\{R^L, R^C,R^V\}$, we define three algorithmic problems, 
which we call $(\calS, R)$-{\sc ExistNE}, $(\calS, R)$-{\sc TieNE},
and $(\calS, R)$-{\sc SingleNE}.
In each of these problems, we are given a candidate set $C$, $|C|=m$, 
a voter set $N$, $|N|=n$, and a utility vector $\vecu=(u_1, \dots, u_n)$, 
where each $u_i$ is represented by $m$ numbers $u_i(c_1), \dots, u_i(c_m)$;
these numbers are positive integers given in binary. In 
$(\calS, R)$-{\sc TieNE} and $(\calS, R)$-{\sc SingleNE} we are also given the name 
of a target candidate $c_p\in C$.
In $(\calS, R)$-{\sc ExistNE} we ask if $(\calS, R, \vecu)$
has a PNE. In $(\calS, R)$-{\sc TieNE} we ask if $(\calS, R, \vecu)$
has a PNE $\vecb$ with $|W(\vecb)|>1$ and $c_p\in W(\vecb)$.
In $(\calS, R)$-{\sc SingleNE} we ask if $(\calS, R, \vecu)$
has a PNE $\vecb$ with $W(\vecb)=\{c_p\}$.
Each of these problems is obviously in NP, as we can simply guess an appropriate 
ballot vector $\vecb$ and check that it is a PNE.

We omit some of the proofs due to space constraints;
these proofs can be found in the supplementary material.

\section{Lazy Voters}\label{sec:lazy}
\noindent
In this section, we study PNE in Plurality games with lazy voters.
The case where the tie-breaking rule is $R^C$ has been analyzed in detail
by Desmedt and Elkind~\cite{des-elk:c:eq}, 
albeit for a slightly different model; we complement their results by considering 
$R^L$ and~$R^V$.

We start by extending a result of Desmedt and Elkind
to all three tie-breaking rules considered in this paper.
\begin{proposition}\label{prop:lazy-basic}
For every $R\in\{R^L,R^C, R^V\}$ and every utility profile $\vecu$,
if a ballot vector $\vecb$ is a PNE of $(\calL,R,\vecu)$
then for every voter $i\in N$ either $b_i=\bot$ or $b_i\in W(\vecb)$.
Further, if $|W(\vecb)|=1$, then there exists exactly one voter $i\in N$
with $b_i\neq\bot$.
\end{proposition}
\begin{proof}
Suppose 
that $b_i\not\in W(\vecb)$
for some voter $i\in N$. Then if $i$ changes her vote to $\bot$,
the set $W(\vecb)$ will not change,
so $i$'s utility would improve by $\eps$, a contradiction with $\vecb$
being a PNE  of $(\calL,R,\vecu)$.
Similarly, suppose that 
$|W(\vecb)|=1$ and there are two voters $i,i'\in N$
with $b_i\neq\bot$, $b_{i'}\neq\bot$.
It has to be the case that $b_i=b_{i'}=c_j$ for some $c_j\in C$, 
since otherwise $|W(\vecb)|\ge 1$. But then if voter $i$
changes her vote to $\bot$, $c_j$ will remain the election winner,
so $i$'s utility would improve by $\eps$, a contradiction.
\end{proof}

\smallskip

\noindent{\bf Lexicographic Tie-breaking\quad} 
The scenario where voters are lazy and
ties are broken lexicographically
turns out to be fairly easy to analyze.

\begin{theorem}\label{thm:lazy-lex-char}
For any utility profile $\vecu$ the game $G = (\calL, R^L, \vecu)$
has the following properties:
\begin{enumerate}
\item
If $\vecb$ is a PNE of $G$ then $|W(\vecb)|\in\{1,m\}$.
Moreover, $|W(\vecb)|=m$ if and only if $\vecb$ is the trivial ballot 
and all voters rank $c_1$ first. 
\item
If $\vecb$ is a PNE of $G$ then
there exists at most one voter $i$ with $b_i\neq \bot$.
\item
$G$ admits a PNE if and only if all voters 
rank $c_1$ first (in which case $c_1$ is the unique PNE winner)
or there exists a candidate $c_j$ with $j>1$ such that (i) $\score(c_j,\veca)>0$
and (ii) for every $k<j$ it holds that all voters prefer $c_j$ to $c_k$.
If such a candidate exists, he is unique, and wins in all PNE of $G$.
\end{enumerate}
\end{theorem}

The following corollary is directly implied by Theorem~\ref{thm:lazy-lex-char}.
 
\begin{corollary}\label{cor:lazy-lex-easy}
$(\calL, R^L)$-{\sc ExistNE}, $(\calL, R^L)$-{\sc SingleNE} and
$(\calL, R^L)$-{\sc TieNE}
are in~{\em P}.
\end{corollary}

\begin{remark}\label{rem:lazy-lex}
{\em
The reader may observe that, counterintuitively, while the lexicographic tie-breaking
rule appears to favor $c_1$, it is impossible for $c_1$ to win the election
unless he is ranked first by all voters. In contrast, $c_2$ wins the election
as long as he is ranked first by at least one voter and no voter prefers $c_1$ to $c_2$.
In general, the lexicographic tie-breaking rule favors lower-numbered candidates
with the exception of $c_1$. As for $c_1$, his presence mostly has a destabilizing effect:
if some, but not all voters rank $c_1$ first, no PNE exists. 
%
This phenomenon is an artifact of our treatment of the trivial ballot vector:
it disappears if we assume (as Desmedt and Elkind do) that when $\vecb=(\bot,\dots,\bot)$
the election is declared invalid and the utility of each voter is $-\infty$: 
under this assumption $c_1$ is the unique possible equilibrium winner whenever he
is ranked first by at least one voter.  
}
\end{remark}

\smallskip

\noindent{\bf Randomized Tie-breaking\quad }
We will now consider 
$R^C$ and $R^V$. 
\cite{des-elk:c:eq} 
characterize utility profiles that admit a PNE for lazy voters and $R^C$. 
However, there is a small difference
between our model and that of Desmedt and Elkind: while we assume that
the trivial ballot vector results in a tie among all candidates, Desmedt and Elkind
assume that in this case the election is canceled and each voter's utility
is $-\infty$. Further, the results of Desmedt and Elkind implicitly assume
that the number of voters $n$ exceeds the number of candidates $m$;
if this is not the case, Theorem~2 in their paper is incorrect
(see Remark~\ref{rem:corr}).

Thus, we will now provide a full characterization of utility profiles $\vecu$
such that $(\calL, R^C,\vecu)$ admits a PNE, and describe the corresponding  
equilibrium ballot profiles. 
Our characterization  
result remains essentially unchanged if we replace $R^C$ with $R^V$: 
for almost all utility profiles $\vecu$ and ballot vectors $\vecb$
it holds that $\vecb$ is a PNE of $(\calL, R^C,\vecu)$ if and only if it is a PNE
of $(\calL, R^V,\vecu)$; the only exception is the case of full consensus
(all voters rank the same candidate first).

\begin{theorem}\label{thm:char-rand}
Let $\vecu=(u_1, \dots, u_n)$ be a utility profile over $C$, $|C|=m$,
and let $R\in\{R^C, R^V\}$.
The game $G = (\calL, R, \vecu)$ admits a PNE
if and only if one of the following conditions holds:
\begin{itemize}
\item[(1)]
all voters rank some candidate $c_j$ first;
\item[(2)]
each candidate is ranked first by at most one voter, and, moreover,
$\frac{1}{n}\sum_{i\in N}u_\ell(a_i)\ge \max_{i\in N\setminus\{\ell\}}u_\ell(a_i)$
for each $\ell\in N$.
\item[(3)]
there exists a set of candidates
$X = \{c_{\ell_1}, \dots, c_{\ell_k}\}$ with $2\le k \le \min(n/2,m)$
and a partition of the voters
into $k$ groups $N_1, \dots, N_k$ of size ${n}/{k}$ each such that for each
$j\in[k]$ and each $i\in N_j$ we have $c_{\ell_j}\succ_i c$
for all $c\in X\setminus\{c_{\ell_j}\}$, and, moreover, 
$\frac{1}{k}\sum_{c\in X}u_i(c)\ge \max_{c\in X\setminus\{c_{\ell_j}\}}u_i(c)$.
\end{itemize}
Further, if condition (1) holds for some $c_j\in C$,
then if $R=R^C$ then for each $i\in N$ the game $G$ has a PNE
where $i$ votes for $c_j$ and all other voters abstain, 
whereas if $R=R^V$ the game $G$ has a PNE where all voters abstain;
if condition (2) holds, then $G$ has a PNE
where each voter votes for her top candidate;
and
if condition (3) holds for some set $X$, then $G$ has a PNE
where each voter votes for her favorite candidate in $X$.
The game $G$ has no other PNE.
\end{theorem}

\begin{remark}\label{rem:corr}
{\em
Desmedt and Elkind claim (Theorems 1 and 2) that 
for $R^C$ and lazy voters 
a PNE exists
if and only if the utility profile satisfies either condition (1)
or condition (3) with constraint $k\le n/2$ removed. To see why this
is incorrect, consider a $2$-voter election 
over the candidate set $C=\{x,y, z\}$, where voters' utility functions are consistent
with preference orders $x\succ y\succ z$ and $x\succ z\succ y$, respectively.
According to Desmedt and Elkind, the ballot vector $(y, z)$ is a PNE of 
the corresponding game. This is obviously not true:
each of the voters would prefer to change her vote to $x$.
Note, however, that the two characterizations differ only when $m\ge n$, 
and in practice the number of voters usually exceeds the number of candidates. 
}
\end{remark}

Desmedt and Elkind show that checking condition (3) of Theorem~\ref{thm:char-rand}
is NP-hard; in their proof $n>m$, and the proof does not depend on how the trivial
ballot is handled. Further, their proof shows that checking whether
a given candidate belongs to some such set $X$ is also NP-hard.
On the other hand, Theorem~\ref{thm:char-rand} shows that PNE
with singleton winning sets only arise if some candidate is unanimously ranked first, 
and this condition is easy to check. 
We summarize these observations as follows.

\begin{corollary}\label{cor:lazy-rand-hard}
For $R\in\{R^C, R^V\}$, the problems
$(\calL, R)$-{\sc ExistNE} and $(\calL, R)$-{\sc TieNE}
are {\em NP}-complete, whereas
$(\calL, R)$-{\sc SingleNE} is in {\em P}.
\end{corollary}

\section{Truth-biased Voters}\label{sec:truth}
\noindent
For truth-biased voters, our exposition follows the same pattern as for lazy voters:
we present some general observations, followed by a quick summary of the results
for lexicographic tie-breaking, and conclude by analyzing randomized tie-breaking.
The following result is similar in spirit to Proposition~\ref{prop:lazy-basic}.

\begin{proposition}\label{prop:truth-basic}
For every $R\in\{R^L,R^C, R^V\}$ and every utility profile $\vecu$,
if a ballot vector $\vecb$ is a PNE of $(\calT,R,\vecu)$
then for every voter $i\in N$ either $b_i=a_i$, or $b_i\in W(\vecb)$.
\end{proposition}
\begin{proof}
Consider a voter $i\in N$ such that $a_i\neq b_i$ and $b_i\not\in W(\vecb)$. 
Suppose $a_i\not\in W(\vecb)$. Then, if $i$ changes her vote to $a_i$, the new winning set 
is either $W(\vecb)$ or $W(\vecb)\cup\{a_i\}$. In either case, $i$'s
utility increases at least by $\eps$, a contradiction. Suppose now that $a_i\in W(\vecb)$. This means that either $W(\vecb) = \{a_i\}$ or $a_i$ 
is in a tie with other candidates under $\vecb$. Then, if $i$ votes for $a_i$, the new winning set is just 
$\{a_i\}$, so $i$'s utility increases by at least $\eps$, a contradiction again.
\end{proof}


\smallskip

\noindent{\bf Lexicographic Tie-breaking\quad}
Obraztsova et al.~\cite{obr-mar-tho:c:truth-biased} 
characterize the PNE of the game $(\calT, R^L, \vecu)$.
As their characterization is quite complex, 
we will not reproduce it here. However, for the purposes of comparison 
with the lazy voters model, we will use the following description
of {\em truthful} equilibria given by Obraztsova et al.

\begin{proposition}[Obraztsova et al., Theorem 1]\label{prop:truth-lex}
Consider a utility profile $\vecu$, let $\veca$ be the respective truthful ballot vector, 
and let $j=\min\{r\mid c_r\in  W(\veca)\}$.
Then $\veca$ is a PNE of $(\calT, R^L, \vecu)$
if and only if neither of the following conditions holds:
\begin{itemize}
\item[(1)]
$|W(\veca)|>1$,
and there exists a candidate $c_k\in W(\veca)$
and a voter $i$ such that $a_i\neq c_k$ and $c_k\succ_i c_j$.
\item[(2)]
$H(\veca)\neq\emptyset$, and there exists a candidate $c_k\in H(\veca)$
and a voter $i$ such that $a_i\neq c_k$, $c_k\succ_i c_j$, and $k<j$.
\end{itemize}
\end{proposition}

We will also state a crucial property of non-truthful PNE, identified by Obraztsova et al. 
For this, we first need the following definition.

\begin{definition}
\label{def:threshold}
Consider a ballot vector $\vecb$, where candidate $c_j$ 
is the winner under $R^L$. A candidate $c_k\neq c_j$ is called a 
{\em threshold 
candidate with respect 
to $\vecb$} if either (1) $k< j$ and $\score(c_k,\vecb)=\score(c_j,\vecb)-1$ or
               (2) $k> j$ and $\score(c_k,\vecb)=\score(c_j,\vecb)$.
We denote the set of threshold candidates with respect to $\vecb$ by $T(\vecb)$.
\end{definition}

That is, a threshold candidate is someone who could win the election if he had one additional vote. 
A feature of all non-truthful PNE is that there must exist at least one threshold candidate. The intuition for this is 
that, since voters who are not pivotal prefer to vote truthfully, in any PNE that arises under strategic voting,  
the winner receives just enough votes so as to beat the required threshold 
(as set by the threshold candidate) and not any more.

\begin{lemma}[Obraztsova et al., Lemma 2]
\label{lem:threshold}
Consider a utility profile $\vecu$, let $\veca$ be the respective truthful ballot vector, and
let $\vecb\neq\veca$ be a non-truthful PNE 
of $(\calT, R^L, \vecu)$. Then $T(\vecb)\neq\emptyset$. 
Further, $\score(c_k, \vecb)=\score(c_k, \veca)$
for every $c_k\in T(\vecb)$,
i.e., all voters whose top choice is $c_k$ vote 
for $c_k$.
\end{lemma}

The existence of a threshold candidate is an important observation about the structure of non-truthful PNE,  
and we will use it repeatedly in the sequel. We note that the winner in $\veca$ 
need not necessarily be a threshold candidate in a non-truthful PNE $\vecb$. 

Obraztsova et al.~show that, given 
a candidate $c_p\in C$ and a score $s$, it is computationally hard to decide whether the game 
$(\calT, R^L, \vecu)$ has a PNE $\vecb$ where $c_p$ wins 
with a score of $s$. This problem may appear to be ``harder''
than $(\calT, R^L)$-{\sc TieNE} or $(\calT, R^L)$-{\sc SingleNE},
as one needs to ensure that $c_p$ obtains a specific score;
on the other hand, it does not distinguish between 
$c_p$ being the unique top-scorer or being tied with other candidates and 
winning due to tie-breaking. We now complement
this hardness result by showing that
all three problems we consider are NP-hard for $\calT$ and $R^L$.

\begin{theorem}\label{thm:truth-lex-hard}
$(\calT, R^L)$-{\sc SingleNE},
$(\calT, R^L)$-{\sc ExistNE}, and $(\calT, R^L)$-{\sc TieNE}
are {\em NP}-complete.
\end{theorem}

The proof is by a reduction
from {\sc Maximum $k$-Subset Intersection (MSI)} (see the supplementary material).
Surprisingly, the complexity of MSI was very recently posed as an open problem  
by Clifford and Popa~\cite{cli-pop:j:subset}; 
subsequently, MSI was shown to be hard under Cook reductions by
Xavier~\cite{xav:j:subset}. 
Here we first establish NP-hardness of MSI under Karp reductions, 
which may be of independent interest, and then show NP-hardness of our problems 
by constructing reductions from MSI. 

\smallskip

\noindent{\bf Randomized Tie-breaking\ }
It turns out that
for truth-biased voters the tie-breaking rules
$R^C$ and $R^V$ induce identical behavior by the voters;
unlike for lazy voters, this holds
even if all voters rank the same candidate first.

For clarity, we present our characterization result
for randomized tie-breaking in three parts. 
We start by considering PNE with winning sets of size at least $2$;
the analysis for this case turns out to be very similar to  
that for lazy voters.

\begin{theorem}\label{thm:truth-rand-ties}
Let $\vecu=(u_1, \dots, u_n)$ be a utility profile over $C$, $|C|=m$,
and let $R\in\{R^C, R^V\}$.
The game $G = (\calT, R, \vecu)$ admits a PNE with a winning set of size at least $2$
if and only if one of the following conditions holds:
\begin{itemize}
\item[(1)]
each candidate is ranked first by at most one voter, and, moreover,
$\frac{1}{n}\sum_{i\in N}u_\ell(a_i)\ge \max_{i\in N\setminus\{\ell\}}u_\ell(a_i)$
for each $\ell\in N$.
\item[(2)]
there exists a set of candidates
$X = \{c_{\ell_1}, \dots, c_{\ell_k}\}$ with $2\le k \le \min(n/2,m)$
and a partitioning of the voters
into $k$ groups $N_1, \dots, N_k$, of size ${n}/{k}$ each, such that for each
$j\in[k]$ 
and each $i\in N_j$, we have $c_{\ell_j}\succ_i c$
for all $c\in X\setminus\{c_{\ell_j}\}$, and, moreover,
$\frac{1}{k}\sum_{c\in X}u_i(c)\ge \max_{c\in X\setminus\{c_{\ell_j}\}}u_i(c)$.
\end{itemize}
Further, 
if condition (1) holds, then $G$ has a PNE
where each voter votes for her top candidate,
and
if condition (2) holds for some $X$, then $G$ has a PNE
where each voter votes for her favorite candidate in $X$.
The game $G$ has no other PNE.
\end{theorem}

The case where the winning set is a singleton is surprisingly complicated.
We will first characterize utility profiles that admit a truthful PNE
with this property.

\begin{theorem}\label{thm:truth-rand-single1}
Let $\vecu=(u_1, \dots, u_n)$ be a utility profile over $C$,
let $R\in\{R^C, R^V\}$, and suppose that $W(\veca)=\{c_j\}$ for some $c_j\in C$.
Then $\veca$ is a PNE of the game $G = (\calT, R, \vecu)$
if and only if for every $i\in N$ and every $c_k\in H(\veca)\setminus\{a_i\}$,
it holds that $c_j\succ_i c_k$.
\end{theorem}

Finally, we consider elections that have non-truthful equilibria with 
singleton winning sets.
\begin{theorem}\label{thm:truth-rand-single2}
Let $\vecu=(u_1, \dots, u_n)$ be a utility profile over $C$,
let $R\in\{R^C, R^V\}$, and consider a ballot vector $\vecb$
with $W(\vecb)=\{c_j\}$ for some $c_j\in C$ and $b_r\neq a_r$ for some $r\in N$.
Then $\vecb$ is a PNE of the game $G = (\calT, R, \vecu)$
if and only if all of the following conditions hold:
\begin{itemize}
\item[(1)]
$b_i\in\{a_i, c_j\}$ for all $i\in N$;
\item[(2)]
$H(\vecb)\neq\emptyset$;
\item[(3)]
$c_j\succ_i c_k$ for all $i\in N$ and all $c_k\in H(\vecb)\setminus\{b_i\}$;
\item[(4)]
for every candidate $c_\ell\in H'(\vecb)$ 
and each voter $i\in N$ with $b_i=c_j$,
$i$ prefers $c_j$ to the lottery where a candidate
is chosen from $H(\vecb)\cup\{c_j, c_\ell\}$ according to $R$. 
\end{itemize}
\end{theorem}

We now consider the complexity of {\sc ExistNE}, {\sc TieNE},
and {\sc SingleNE} for truth-biased voters and randomized tie-breaking.
The reader may observe that the characterization of PNE with ties
in Theorem~\ref{thm:truth-rand-single2} is essentially identical 
to the one in Theorem~\ref{thm:char-rand}. As a consequence,
we immediately obtain that $(\calT, R^C)$-{\sc TieNE}
and $(\calT, R^V)$-{\sc TieNE} are NP-hard.
For {\sc ExistNE} and {\sc SingleNE}, a simple modification
of the proof of Theorem~\ref{thm:truth-lex-hard} shows that
these problems remain hard under randomized tie-breaking.
These observations are summarized in the following corollary.

\begin{corollary}\label{cor:truth-rand-hard}
For $R\in\{R^C, R^V\}$, $(\calT, R)$-{\sc SingleNE},
$(\calT, R)$-{\sc TieNE}, and $(\calT, R)$-{\sc ExistNE}
are {\em NP}-complete.
\end{corollary}
\section{Comparison}\label{sec:comparison}
\noindent
We are finally in a position to compare the different models considered in this paper. 

\noindent{\bf Tie-breaking rules\ }
We have demonstrated that in equilibrium
the two randomized tie-breaking rules ($R^C$ and $R^V$)
induce very similar voter behavior, and identical election outcomes, both for lazy and
for truth-biased voters. This is quite remarkable, since under truthful
voting these tie-breaking rules can result in very different lotteries.
In contrast, there is a substantial
difference between the randomized rules and the lexicographic rule.
For instance, when voters are lazy, {\sc ExistNE} is NP-hard for $R^C$ and $R^V$,
but polynomial-time solvable for $R^L$. Further, the lexicographic rule
is, by definition, not anonymous, and Theorem~\ref{thm:lazy-lex-char} demonstrates
that candidates with smaller indices have a substantial advantage.
For truth-biased voters the impact of tie-breaking rules is less clear:
while we have obtained NP-hardness results for all three rules,
it appears that, in contrast with lazy voters, for truth-biased voters
randomized tie-breaking induces  ``simpler'' PNE than lexicographic tie-breaking.

\noindent{\bf Lazy vs. truth-biased voters\ }
Under lexicographic tie-breaking,
the sets of equilibria induced by the two types of secondary preferences
are incomparable: there exists a utility profile $\vecu$ such that 
the sets of candidates who can win in PNE of $(\calL, R^L, \vecu)$ 
and $(\calT,R^L, \vecu)$ are disjoint.
\begin{example}
{\em
Let $C=\{c_1,c_2, c_3\}$, and consider a $4$-voter election with
one vote of the form $c_2\succ c_3\succ c_1$, and
three votes of the form $c_3\succ c_2\succ c_1$.
The only PNE of $(\calL, R^L, \vecu)$ is $(c_2, \bot, \bot, \bot)$, 
where $c_2$ wins, whereas the only 
PNE of $(\calT, R^L, \vecu)$ is $(c_2, c_3, c_3, c_3)$, 
where $c_3$ wins.
}
\end{example}

For randomized tie-breaking, the situation is more interesting. For concreteness, 
let us focus on $R^C$. Note first that
the utility profiles for which there exist PNE with winning sets of size $2$ or more
are the same for both voter types. Further, if $(\calL, R^C, \vecu)$
has a PNE $\vecb$ with $|W(\vecb)|=1$ (which happens only if there
is a unanimous winner), then $\vecb$ is also a PNE
of $(\calT, R^C, \vecu)$. However, $(\calT, R^C, \vecu)$ may have additional
PNE, including some non-truthful ones. In particular,
for truth-biased voters,
the presence of a strong candidate is sufficient for stability: 
Proposition~\ref{prop:truth-lex} implies that
if there exists a $c\in C$ such that $\score(c, \veca)\ge \score(c',\veca)+2$ for  
all $c'\in C\setminus\{c\}$, then for any $R\in\{R^L, R^C, R^V\}$
the ballot vector $\veca$ is a PNE of $(\calT, R, \vecu)$ with $W(\veca)=\{c\}$.

\noindent{\bf Existence of PNE\ }
One can argue that, when the number of voters is large relative to the number of candidates, 
under reasonable probabilistic models of elections,   
the existence of a strong candidate (as defined in the previous paragraph)
is exceedingly likely (we omit the formal statement
of this result and its proof due to space constraints), so elections
with truth-biased voters typically admit stable outcomes; 
this is corroborated by the experimental 
results of \cite{tho-lev-ley:c:empirical}.
In contrast,
for lazy voters stability is more difficult to achieve, unless there is a candidate
that is unanimously ranked first: under randomized tie-breaking rules,
there needs to be a very precise balance among the candidates that end
up being in $W(\vecb)$, and under $R^L$ the eventual winner
has to Pareto-dominate all candidates that lexicographically precede him.  
Either of these conditions appears to be quite difficult to satisfy in a large election.

\noindent{\bf Quality of PNE\ }
In all of our models,
a candidate ranked last by all voters cannot be elected, in contrast to the basic game-theoretic model for Plurality voting. 
However, not all non-desirable outcomes are eliminated: 
under $R^V$ and $R^C$ both lazy voters and truth-biased voters can still elect
a Pareto-dominated candidate with non-zero probability in PNE. 
This has been shown for lazy voters and $R^C$ by~\cite{des-elk:c:eq} (Example 1),
and the same example works for truth-biased voters and for $R^V$.
A similar construction shows that a Pareto-dominated candidate
may win under $R^L$ when voters are truth-biased. 
%
In contrast, lazy voters cannot elect a Pareto-dominated candidate under $R^L$:
Theorem~\ref{thm:lazy-lex-char} shows that the winner has to be ranked
first by some voter. 

We can also measure the quality of PNE 
by analyzing the Price of Anarchy (PoA) in both models.
The study of PoA in the context of 
voting has been recently initiated by Branzei et al.~\cite{bcmp_2013}. 
The additive version of PoA, which was considered by Branzei et al., 
is defined as the worst-case difference between the score of the winner under truthful
voting and the truthful score of a PNE winner. 
It turns out that PoA can be quite high, both for lazy and truth-biased voters.
To illustrate this, we provide in the supplementary material two examples 
showing that under lexicographic tie-breaking $\PoA = \Omega (n)$ in both models. 
Similar results can be established for randomized tie-breaking as well.

Even though the $\PoA$ results are not encouraging, 
this is only a worst-case analysis and we expect PNE to have a better performance on average.
For the truth-biased model, this is also supported by the experimental evaluation of 
Thompson et al.~\cite{tho-lev-ley:c:empirical}, 
who showed that in the truth-biased model 
most PNE identified in their simulations had good social welfare properties. 
Formalizing this observation, i.e., providing
average-case analysis of the quality of PNE in voting games, 
is a promising topic for future work.


\section{Extension: Principled Voters}\label{sec:principled}
\noindent
The results of this paper can be extended to the setting where some of the voters
are {\em principled}, i.e., always vote truthfully (and never abstain). 
%
Due to space constraints, we 
relegate the formal statements of our results for this extended model to the supplementary material.
Briefly, the presence of principled voters has the strongest effect 
on lazy voters and lexicographic tie-breaking, whereas for other settings the effect
is less pronounced. All computational problems that were easy in the standard model
remain easy in the extended model (and, obviously, all hard problems remain hard).
Finally, in the presence of principled voters the random candidate tie-breaking
rule is no longer equivalent to the random voter tie-breaking rule.

\section{Conclusions and Future Work}
\noindent
We have characterized PNE of Plurality voting for several combinations of secondary 
preferences and tie-breaking rules. Our complexity results are summarized 
in Table~\ref{tbl:summary}. A promising direction for future work 
is to investigate more general classes of tie-breaking rules.
It is also interesting to consider the complexity of
various refinements of Nash equilibria for our models,
such as strong Nash equilibria (for which an analysis for $\calT$ and $R^L$ 
can be found in the work of Obraztsova et al.~\cite{obr-mar-tho:c:truth-biased}), 
or subgame-perfect Nash equilibria for settings where
voters submit their ballots one by one; 
see \cite{des-elk:c:eq} and \cite{xia-con:c:spne} for some results about
such equilibria.

\begin{table}[ht]
{\small
\begin{tabular}{|r|c|c|c|}
\hline
 & {\sc SingleNE}    & {\sc TieNE}	& {\sc ExistNE}	\\
\hline
  $(\calL, R^L)$ & P (Cor.~\ref{cor:lazy-lex-easy})            	& P (Cor.~\ref{cor:lazy-lex-easy})		& P (Cor.~\ref{cor:lazy-lex-easy})		\\
\hline
  $(\calL, R^C)$ & P (Cor.~\ref{cor:lazy-rand-hard})            & NPc (Cor.~\ref{cor:lazy-rand-hard})          & NPc (Cor.~\ref{cor:lazy-rand-hard})          \\
\hline
  $(\calL, R^V)$ & P (Cor.~\ref{cor:lazy-rand-hard})            & NPc (Cor.~\ref{cor:lazy-rand-hard}) & NPc (Cor.~\ref{cor:lazy-rand-hard}) \\
\hline
  $(\calT, R^L)$ & NPc (Thm.~\ref{thm:truth-lex-hard})                 & NPc (Thm.~\ref{thm:truth-lex-hard})           & NPc (Thm.~\ref{thm:truth-lex-hard})          \\
\hline
  $(\calT, R^C)$ & NPc (Cor.~\ref{cor:truth-rand-hard})                 & NPc (Cor.~\ref{cor:truth-rand-hard})          & NPc (Cor.~\ref{cor:truth-rand-hard})          \\
\hline
  $(\calT, R^V)$ & NPc (Cor.~\ref{cor:truth-rand-hard})                 & NPc (Cor.~\ref{cor:truth-rand-hard})          & NPc (Cor.~\ref{cor:truth-rand-hard})           \\
\hline
\end{tabular}
}
\caption{\label{tbl:summary}Complexity results: P stands for ``polynomial-time solvable'', 
NPc stands for ``NP-complete''.}
\end{table}


\newpage


\newpage

\appendix
\section{Proofs Omitted from Section 3}

\smallskip

\noindent{\bf Theorem \ref{thm:lazy-lex-char}.\ }
{\em 
For any utility profile $\vecu$ the game $G = (\calL, R^L, \vecu)$
has the following properties:
\begin{enumerate}
\item
If $\vecb$ is a PNE of $G$ then $|W(\vecb)|\in\{1,m\}$.
Moreover, $|W(\vecb)|=m$ if and only if $\vecb$ is the trivial ballot 
and all voters rank $c_1$ first. 
\item
If $\vecb$ is a PNE of $G$ then
there exists at most one voter $i$ with $b_i\neq \bot$.
\item
$G$ admits a PNE if and only if all voters 
rank $c_1$ first (in which case $c_1$ is the unique PNE winner)
or there exists a candidate $c_j$ with $j>1$ such that (i) $\score(c_j,\veca)>0$
and (ii) for every $k<j$ it holds that all voters prefer $c_j$ to $c_k$.
If such a candidate exists, he is unique, and wins in all PNE of $G$.
\end{enumerate}
}

\begin{proof}
Fix a utility profile $\vecu$ and a ballot $\vecb$
such that $\vecb$ is a PNE of $G = (\calL, R^L, \vecu)$.

To prove the first claim,
suppose first that  $1<|W(\vecb)|$ and $\vecb$ is not trivial. Then 
there are two candidates $c_j, c_k\in W(\vecb)$, $j<k$, such that
$\score(c_j, \vecb)>0$ and $\score(c_k,\vecb)>0$.  
Hence, there exists at least one voter who votes for $c_k$. However, 
the election outcome will not change if this voter abstains, 
a contradiction with $\vecb$ being a PNE of $G$.  
Now, suppose that $\vecb$ is trivial. In this case $W(\vecb)=C$
and $c_1$ wins. If any voter prefers some other candidate $c$ to $c_1$,
she can improve her utility by voting for $c$, as this will change the election
outcome to $c$. On the other hand, if all voters rank $c_1$ first, the trivial ballot
is clearly a PNE. 

The second claim follows from our first claim 
and Proposition~\ref{prop:lazy-basic}.

To prove the third claim, suppose that there exists a candidate $c_j$, $j > 1$, satisfying conditions (i) and (ii).
Consider a ballot vector $\vecb$ where $b_i=c_j$ for some voter $i$ with $a_i=c_j$
(the existence of such voter is guaranteed by condition (i))
and $b_{i'}=\bot$ for all $i'\in N\setminus\{i\}$. Voter $i$ cannot benefit
from voting for another candidate or abstaining, as this will change 
the election outcome to one she likes less than the current outcome.
Any other voter can only change the election outcome if she votes
for a candidate $c_k$ with $k<j$. But then condition (ii) implies
that no voter wants the election outcome to change in this way.
Conversely, suppose that $\vecb$ is a PNE. We have argued that
either $\vecb$ is trivial or 
$b_i=c_j$ for some $i\in N$ and some $c_j\in C$ and $b_{i'}=\bot$ 
for all $i'\in N\setminus\{i\}$. 
In the latter case, if $c_j\neq a_i$, voter $i$ can improve her utility
by voting for $a_i$. Moreover, if $j=1$, voter $i$
can improve her utility by abstaining, as $c_1$ would remain
the election winner in this case. Finally, if there exists 
a candidate $c_k$ with $k<j$ such that some voter $i'$ prefers
$c_k$ to $c_j$, then $i'$ can change the election outcome to $c_k$
by voting for $c_k$. 

It remains to show that conditions (i) and (ii) can be satisfied by at most one candidate.
To see this, note that if both $c_j$ and $c_k$ satisfy condition (i) and 
$j<k$, then $c_k$ violates condition (ii), as the voter who ranks $c_j$
first clearly prefers $c_j$ to $c_k$.  
\end{proof}

\smallskip

\noindent{\bf Theorem \ref{thm:char-rand}.\ }
{\em
Let $\vecu=(u_1, \dots, u_n)$ be a utility profile over $C$, $|C|=m$,
and let $R\in\{R^C, R^V\}$.
The game $G = (\calL, R, \vecu)$ admits a PNE
if and only if one of the following conditions holds:
\begin{itemize}
\item[(1)]
all voters rank some candidate $c_j$ first;
\item[(2)]
each candidate is ranked first by at most one voter, and, moreover,
$\frac{1}{n}\sum_{i\in N}u_\ell(a_i)\ge \max_{i\in N\setminus\{\ell\}}u_\ell(a_i)$
for each $\ell\in N$.
\item[(3)]
there exists a set of candidates
$X = \{c_{\ell_1}, \dots, c_{\ell_k}\}$ with $2\le k \le \min(n/2,m)$
and a partition of the voters
into $k$ groups $N_1, \dots, N_k$ of size ${n}/{k}$ each such that for each
$j\in[k]$ and each $i\in N_j$ we have $c_{\ell_j}\succ_i c$
for all $c\in X\setminus\{c_{\ell_j}\}$, and, moreover, 
$\frac{1}{k}\sum_{c\in X}u_i(c)\ge \max_{c\in X\setminus\{c_{\ell_j}\}}u_i(c)$.
\end{itemize}
Further, if condition (1) holds for some $c_j\in C$,
then if $R=R^C$ then for each $i\in N$ the game $G$ has a PNE
where $i$ votes for $c_j$ and all other voters abstain, 
whereas if $R=R^V$ the game $G$ has a PNE where all voters abstain;
if condition (2) holds, then $G$ has a PNE
where each voter votes for her top candidate;
and
if condition (3) holds for some set $X$, then $G$ has a PNE
where each voter votes for her favorite candidate in $X$.
The game $G$ has no other PNE.
}

\begin{proof}
It is easy to see that any of the conditions (1)--(3) is sufficient for the existence
of PNE, with ballot vectors described in the statement of the theorem
witnessing this. We will now show that satisfying at least one of these conditions is necessary
for the existence of a PNE, and that no other ballot vector is a PNE.
Fix a tie-breaking rule $R\in\{R^C, R^V\}$, a utility profile $\vecu$, 
and suppose that a ballot vector $\vecb$ is a PNE of $(\calL, R,\vecu)$.
We will argue that $\vecu$ satisfies one of the conditions (1)--(3).

Suppose first that $W(\vecb)=\{c_j\}$ for some $c_j\in C$.
By Proposition~\ref{prop:lazy-basic} there exists a voter $i\in N$
with $b_i=c_j$, and $b_{i'}=\bot$ for all $i'\in N\setminus\{i\}$.
It has to be the case that $a_i=c_j$: otherwise voter $i$ 
can make $a_i$ the unique winner by changing her vote to $a_i$, 
thus increasing her utility. Now, suppose that $a_{i'}\neq c_j$
for some $i'\in N\setminus\{i\}$. If voter $i'$ changes her ballot
to $c_\ell = a_{i'}$, the new winning set is
$\{c_j, c_\ell\}$. Now, if $R=R^C$, 
the overall utility of $i'$ is given by $\frac12(u_{i'}(c_\ell)+u_{i'}(c_j))$,
and if $R=R^V$, 
the overall utility of $i'$ is given by $\lambda u_{i'}(c_\ell)+(1-\lambda)u_{i'}(c_j)$,
where $\lambda\ge \frac{1}{n}$ (this is because voter $i'$ herself
ranks $c_\ell$ above $c_j$). In both cases, $i'$ can increase her
utility by voting $c_\ell$, a contradiction. Hence, it has
to be the case that all voters rank $c_j$ first, i.e., condition~(1) 
is satisfied.

Now, suppose that $|W(\vecb)|> 1$. We will argue that in this case
either all voters abstain or no voter abstains. 
Indeed, suppose that $b_i=\bot$, $b_{\ell}\neq \bot$ for some $i,\ell\in N$,
i.e., each candidate in $W(\vecb)$ receives at least one vote.
If, instead of abstaining, $i$ votes for her most preferred
candidate in $W(\vecb)$, this candidate becomes the unique election winner.
In contrast, under $\vecb$
$i$'s least preferred candidate in $W(\vecb)$ wins with positive probability:
this is immediate for $R=R^C$, and for $R=R^V$ this holds because
for every $c_j\in W(\vecb)$ there exists a voter $i'$ with $b_{i'}=c_j$,
and $c_j$ wins whenever ties are broken according to the preferences of voter ${i'}$.  
Thus, $i$ can improve her utility by changing her vote, a contradiction.
Hence, if $|W(\vecb)|=k$ and $\vecb$ is not trivial, 
each candidate in $W(\vecb)$ receives exactly $n/k$ votes.

In particular, if $|W(\vecb)|=n$ and $\vecb$ is not trivial, 
each candidate in $W(\vecb)$ receives exactly one
vote. We will argue that in this case condition (2) is satisfied. 
We will first prove that $b_i=a_i$ for all $i\in N$. Indeed, suppose that $b_i\neq a_i$
for some $i\in N$, and consider the ballot vector $\vecb'=(\vecb_{-i}, a_i)$. 
If $a_i\in W(\vecb)$, then $W(\vecb')=\{a_i\}$,
whereas under $\vecb$ voter $i$'s least preferred candidate in 
$W(\vecb)$ wins with positive probability.
If $a_i\not\in W(\vecb)$, we have $W(\vecb')=(W(\vecb)\setminus\{b_i\})\cup\{a_i\}$, 
so $U_i(\vecb')=U_i(\vecb)+\frac{1}{n}(u_i(a_i)-u_i(b_i))>U_i(\vecb)$.
In both cases $i$ can increase her overall utility by voting for $a_i$,
a contradiction. Hence, we have $W(\vecb)=\{a_i\mid i\in N\}$.
Thus, under both $R^C$ and $R^V$ the outcome of this election
is a lottery that assigns equal probability to all candidates in $W(\vecb)$.
Now, if any voter prefers her second most preferred candidate in
$W(\vecb)$ to this lottery, she can vote for that candidate, making him
the unique election winner, a contradiction with $\vecb$ being a PNE.
Thus, in this case condition (2) is satisfied. 

Now, suppose that $\vecb$ is not trivial and $|W(\vecb)|=k<n$. We have argued that each candidate
in $W(\vecb)$ receives exactly $n/k$ votes. This means that $k$ divides $n$, 
so in particular $k\le n/2$ and each candidate in $W(\vecb)$ receives
at least two votes. Under both of our tie-breaking rules, each candidate
in $W(\vecb)$ wins with probability $1/k$. Consider a voter $i$.
She can make any candidate in $W(\vecb)\setminus\{b_i\}$ the unique
election winner by voting for him. Since $\vecb$ is a PNE, no voter
wants to change the election outcome in this way; this implies, in particular,
that each voter votes for her favorite candidate in $W(\vecb)$.
Thus, in this case condition (3) is satisfied with $X=W(\vecb)$;
the voters are partitioned into groups according to their votes in $\vecb$.

It remains to consider the case where $\vecb$ is the trivial ballot vector.
When $R=R^C$, $\vecb$ cannot be a PNE: under $\vecb$ the outcome is a uniform lottery over $C$,
and every voter would rather vote for her favorite candidate in order to make him the unique
winner. When $R=R^V$, the outcome is a lottery that assigns a positive probability
to each candidate in $A=\{a_i\mid i\in N\}$.
If $|A|>1$, $\vecb$ is not a PNE: each voter would prefer to vote 
for her favorite candidate in order to make him the unique winner.   
However, if $A$ is a singleton, i.e., all voters rank some candidate $c_j$
first, the trivial ballot vector is a PNE: after all voters abstain, $R^V$
picks a random voter, and this voter selects $c_j$.
\end{proof}

\smallskip

\section{Proofs Omitted from Section 4}
The following problem is very useful in our constructions.

\begin{definition}
An instance of {\sc Maximum $k$-Subset Intersection (MSI)} is
a tuple $(\calE, A_1, \dots, A_m, k, q)$, where $\calE = \{e_1,\dots,e_n\}$ is a finite set of elements,
each $A_i$, $i\in[m]$, is a subset of $\calE$, and $k, q$ are positive integers.
It is a ``yes''-instance if there exist sets $A_{i_1}, \dots, A_{i_k}$
such that  $|\cap_{j\in [k]}A_{i_j}|\ge q$,
and a ``no''-instance otherwise.
\end{definition}

Despite the relevance of MSI to various optimization scenarios, 
see e.g.~\cite{xav:j:subset}, it was only recently shown that this 
problem is hard under a Cook reduction. We provide below a Karp reduction, 
establishing NP-completeness of MSI.

\begin{theorem}
{\sc MSI} is {\em NP}-complete.
\end{theorem} 

\begin{proof}
Trivially MSI is in NP. For hardness, we provide a reduction from the {\sc Balanced Complete Bipartite Subgraph} problem.
An instance of this problem consists of a bipartite graph $G = (V_1, V_2, E)$, and an integer $k$. It is a "yes"-instance if there exist sets $S_1\subseteq V_1$, $S_2\subseteq V_2$, with $|S_1| = |S_2| = k$, such that the subgraph induced by $S_1$ and $S_2$ is a complete bipartite subgraph. This problem is known to be NP-complete, see~\cite{gar-joh:b:NP}.

Consider an instance $I$ of this problem, given by $G$ and some integer $k$. We construct an instance $I'$ of MSI as follows: we let $\calE = V_1$. We also let the elements of $V_2$ correspond to sets. In particular, for every $j\in V_2$, we have a corresponding set $A_j\subseteq \calE$, such that $A_j = \{i\in V_1: (i, j)\in E\}$. We set the parameter $k$ in MSI to be the same as the integer $k$ from $I$. We also set $q=k$. We now claim that $I$ is a "yes"-instance of {\sc Balanced Complete Bipartite Subgraph} if and only if $I'$ is a "yes"-instance of MSI. 

To see this, suppose there exist $S_1\subseteq V_1$, $S_2\subseteq V_2$, such that we have a complete bipartite subgraph induced by $S_1$ and $S_2$. Then take the $k$ sets corresponding to $S_2$. Clearly every element from $S_1$ belongs to all these sets, hence the intersection of these sets is at least $k$. For the reverse direction, suppose there exists a collection of $k$ sets in $I'$ whose intersection is at least $k$. Then take as $S_2$ the $k$ vertices that correspond to these sets. Let also $S_1$ be any $k$-element subset of the intersection. Then obviously, the bipartite graph induced by $S_1$ and $S_2$ is complete.
\end{proof}

We can now prove Theorem~\ref{thm:truth-lex-hard}, utilizing the hardness of MSI.

\smallskip

\noindent {\bf Theorem \ref{thm:truth-lex-hard}.\ }
{\em
$(\calT, R^L)$-{\sc SingleNE},
$(\calT, R^L)$-{\sc ExistNE}, and $(\calT, R^L)$-{\sc TieNE}
are {\em NP}-complete.
}

\begin{proof}
We will first establish the NP-completeness of $(\calT, R^L)$-{\sc SingleNE}, 
and then show how to modify the proof for the other two problems.
It is trivial to show that $(\calT, R^L)$-{\sc SingleNE} is in NP, 
so we focus on showing that it is NP-hard.

We provide a reduction from {\sc MSI}. 
Consider an instance $I$ of the {\sc MSI} problem. 
We can assume that 
for every $e \in \calE$ there exists an index $i$, such that $e \notin A^i$ and $m>n+k+q$ 
(if this is not the case we can add several additional empty sets).

We now construct an instance of our problem. We have $n+3$ candidates, namely,
$C=\calE \cup \{w_1, w_2, w_3\}$, with ties broken according to $e_1>\ldots e_n>w_3 > w_1 > w_2$.
We set $w_2$ to be the target winning candidate, i.e., $c_p:= w_2$. 
Finally, we set $\delta = \frac{1}{6(n+m)}$. 
We will now describe  the voters' preferences and their utility functions
(while the utility functions play no role in this proof, we will use the same construction
in the NP-hardness proof for randomized tie-breaking, see Corollary~\ref{cor:truth-rand-hard}, where they do matter). 
The voters in our instance are split into five blocks as follows.

\begin{itemize}
\item {\bf Block 1:} For every $i\in [m]$ we construct a  
voter $v_{i}$ who ranks the candidates as
$w_3 \pref \calE \setminus A^i \pref w_2 \pref A^i \pref w_1$.
We let $u_{i}$ denote the utility function of $v_{i}$, which we construct as follows. 
We set $u_{i}(w_3) = 1, u_{i}(w_2) = \frac{1}{2}, u_{i}(w_1) = \frac{1}{4}$. 
Further, $v_{i}$ assigns utility of $1-j\delta$ to her $j$-th most preferred candidate in $\calE \setminus A^i$. 
Note that $|\calE \setminus A^i|<n$, so 
these numbers are strictly between $1$ and $1/2$ and they are consistent with the ranking of voter $v_{i}$. 
Finally, $v_{i}$ assigns utility of $1/2 - j\delta$ to her $j$-th most preferred candidate in $A^i$;
these numbers are strictly between $1/2$ and $1/4$. 

\item {\bf Block 2:} We set $s = m-k+3$, 
and we add $s-1$ voters whose preferences are of the form $w_1 \pref w_2 \pref w_3 \pref \calE$.

\item {\bf Block 3:} We add $s-k-(n-q)-1$ voters with preferences of the form
$w_2 \pref w_1 \pref w_3 \pref \calE$.  

\item {\bf Block 4:} For every $e_j \in \calE$, 
we add $s-2$ voters with preferences of the form $e_j \pref w_3 \pref w_2 \pref w_1 
\pref \calE \setminus \{e_j\}$.

\item {\bf Block 5:} A voter with preferences of the form
$ w_3 \pref w_2 \pref \calE \pref w_1$.
\end{itemize}
Each of the voters in blocks 2--5 assigns utility of $1-\delta(j-1)$
to the $j$-th candidate in her ranking.

Let $I\rq{}$ be the constructed instance. We want to establish that $I$ is a ``yes''-instance of the  {\sc MSI} 
problem if and only if $I\rq{}$ is a ``yes''-instance of $(\calT, R^L)$-{\sc SingleNE}.
Suppose first that $I\rq{}$ is a ``yes''-instance of our problem. Then there exists a PNE $\vecb$ with 
$W(\vecb)=\{w_2\}$. We will first establish some properties of $\vecb$.

Let $\veca$ denote the truthful ballot for $I\rq{}$.
We have $\score(w_1,\veca)=s-1$, $\score(w_2,\veca)=s-k-(n-q)-1$, $\score(w_3,\veca) = m+1 = s+k-2$, 
and $\score(e_j,\veca)=s-2$ for every $e_j \in \calE$. 
It follows that $w_2$ is not among the winners in $\veca$. 

We will now argue that in the PNE $\vecb$, $T(\vecb)=\{w_1\}$. 
We know by Lemma~\ref{lem:threshold} that $T(\vecb)\neq\emptyset$.
Since $w_2$ is the winner at $\vecb$, $w_2\not\in T(\vecb)$.
Also, it is easy to see that $w_3 \notin T(\vecb)$. 
Indeed, suppose the contrary.
All voters in Block 4 prefer $w_3$ to $w_2$. 
By Proposition \ref{prop:truth-basic}, 
in $\vecb$ these voters vote either for their top choice or for $W(\vecb) = \{w_2\}$, hence not for $w_3$. 
But if $w_3\in T(\vecb)$, each of these voters would prefer to switch to voting $w_3$, 
a contradiction with $\vecb$ being a PNE.
A similar argument shows that $\calE\cap T(\vecb)=\emptyset$.
Indeed, we assumed that for every $e_{\ell} \in \calE$ there exists $i$ such that $e_{\ell} \notin 
A^i$. Then voter $v_{i}$ from Block 1 prefers 
$e_{\ell}$ to $w_2$, and $e_\ell$ is not her top choice.
By Proposition \ref{prop:truth-basic},
in $\vecb$ voter $v_{i}$ votes for her top choice or for $w_2$, but if 
$e_\ell\in T(\vecb)$, she would prefer to change her vote to $e_\ell$,
a contradiction with $\vecb$ being a PNE.
As we have ruled out all candidates except for $w_1$, it follows that $T(\vecb)=\{w_1\}$,
and hence, by Lemma \ref{lem:threshold}, $\score(w_1, \vecb)=s-1$. 
Then, by the tie-breaking rule, it must be that $\score(w_2, \vecb)=s$. 
Thus, in $\vecb$ candidate $w_2$ receives exactly $k + n-q+1$ non-truthful votes, 
in addition to the votes of his own supporters. 
We also know that the voters from Block 3 keep voting for $w_2$ in $\vecb$, 
and, by Lemma \ref{lem:threshold}, the voters from Block 2 keep voting for $w_1$ in $\vecb$. 
Hence $w_2$ receives the extra $k + n-q+1$ votes in $\vecb$ from Blocks 1, 4 and 5.

We claim that $\score(w_3,\vecb)\leq s-3$.
Indeed, we have $\score(c',\vecb)\le s-2$ for all $c'\in \calE \cup\{w_3\}$ since $T(\vecb)=\{w_1\}$.
Further, if $\score(w_3,\vecb)= s-2$, 
then $\vecb$ would not be a PNE, as some voters from Blocks 1, 4, and 5 
vote for $w_2$, but all of them prefer $w_3$ to $w_2$.
Thus, since the only supporters of $w_3$ are in Block 1 and Block 5, in total, we must have at least $k+1$ voters from Blocks 1 and 5 who vote for $w_2$ in $\vecb$. 
This means that there are at least $k$ 
voters from Block 1, who have deviated to $w_2$. Now,
we pick all sets $A_{i_j}$  for all $j\in [k]$ such that $v_{i_j}$ votes for $w_2$ in $\vecb$.
We will argue that $|\cap_{j\in [k]} A_{i_j}|\ge q$.

To see this, let $\calE'=\{e\in\calE\mid\score(e, \vecb)=s-2\}$.
Note that in $\vecb$ there are at most $n-q$ voters in Block 4 who vote for $w_2$. 
Hence, we have $|\calE'|\ge q$. To complete the proof, we only need to argue that for each $e\in\calE'$ 
we have $e\in A_{i_j}$ for all $j\in [k]$. Indeed, fix some $e\in\calE'$ and some $j\in [k]$.
By our choice of $A_{i_j}$, the corresponding voter $v_{i_j}$ has voted for $w_2$ in $\vecb$. 
Suppose that $v_{i_j}$ prefers $e$ to $w_2$. 
If she changes her vote to $e$, then $e$ becomes the new winner, due to tie-breaking, 
a contradiction with $\vecb$ being a PNE. Thus, it has to be the case that $v_{i_j}$ prefers $w_2$ to $e$, which means that $e\in A_{i_j}$,
as we wanted to prove.
Hence, a ``yes''-instance for $(\calT, R^L)$-{\sc SingleNE}, 
corresponds to a "yes"-instance of {\sc MSI}. 

For the converse direction, suppose that for the instance $I$, there is a collection of sets $A_{i_1}, \dots, A_{i_k}$
with $|\cap_{j\in [k]} A_{i_j}|\ge q$. Let $\calE'=\cap_{j\in [k]} A_{i_j}$.
We identify below a set of voters, $N'$, from the instance $I'$, which allows us to construct an equilibrium profile. 
We include in $N'$ the set $\{v_{i_j}\mid j\in[k]\}$. We also add to $N'$
the voter from Block 5.
Furthermore, for each $e\not\in \calE'$, we add to $N'$ one voter from Block 4 who ranks $e$ first.
Observe that at this point we have $|N'|\le k+(n-q)+1$. If $|N'|<k+(n-q)+1$, 
we pick $n-q+k+1-|N'|$ additional voters from Block 4, corresponding to supporters of elements $e\not\in \calE'$, and add them to $N'$.
Now, consider a ballot vector $\vecb$ where the voters in $N'$ vote in favor of $w_2$, 
and everyone else votes truthfully.
We have $\score(w_2, \vecb)=s$, $\score(w_1, \vecb)=s-1$, 
$\score(w_3, \vecb)=s-3$, $\score(e, \vecb)\le s-3$
for all $e\in \calE\setminus\calE'$, $\score(e, \vecb)= s-2$ for all $e\in\calE'$.
Hence, $w_2$ is the winner, and all non-truthful voters rank $w_2$ above $w_1$ as well 
as above all candidates in $\calE'$ (who could possibly become winners if some voters had an incentive to vote for them).
Thus, $\vecb$ is a PNE.

Finally, we comment on the hardness of the problems $(\calT, R^L)$-{\sc ExistNE} and $(\calT, R^L)$-{\sc TieNE}. 
For $(\calT, R^L)$-{\sc TieNE}, we can make a small modification to the reduction above. Specifically, it 
suffices to switch the tie-breaking order between $w_1$ and $w_2$, and also add one more voter 
to Block 2 in favor of $w_1$. 
For $(\calT, R^L)$-{\sc ExistNE}, hardness is again based on a modification of the reduction above: we can add $k-1$ additional copies of $w_3$ into the profile and change $s$ to $m+2$.
In this case it can be shown that 
$w_1$ is the only possible threshold candidate and hence only $w_2$ can be a winner in a PNE.
We omit further details from this version.
\end{proof}

\smallskip

\noindent{\bf Theorem \ref{thm:truth-rand-ties}.\ }
{\em
Let $\vecu=(u_1, \dots, u_n)$ be a utility profile over $C$, $|C|=m$,
and let $R\in\{R^C, R^V\}$.
The game $G = (\calT, R, \vecu)$ admits a PNE with a winning set of size at least $2$
if and only if one of the following conditions holds:
\begin{itemize}
\item[(1)]
each candidate is ranked first by at most one voter, and, moreover,
$\frac{1}{n}\sum_{i\in N}u_\ell(a_i)\ge \max_{i\in N\setminus\{\ell\}}u_\ell(a_i)$
for each $\ell\in N$.
\item[(2)]
there exists a set of candidates
$X = \{c_{\ell_1}, \dots, c_{\ell_k}\}$ with $2\le k \le \min(n/2,m)$
and a partitioning of the voters
into $k$ groups $N_1, \dots, N_k$, of size ${n}/{k}$ each, such that for each
$j\in[k]$ 
and each $i\in N_j$, we have $c_{\ell_j}\succ_i c$
for all $c\in X\setminus\{c_{\ell_j}\}$, and, moreover,
$\frac{1}{k}\sum_{c\in X}u_i(c)\ge \max_{c\in X\setminus\{c_{\ell_j}\}}u_i(c)$.
\end{itemize}
Further, 
if condition (1) holds, then $G$ has a PNE
where each voter votes for her top candidate,
and
if condition (2) holds for some $X$, then $G$ has a PNE
where each voter votes for her favorite candidate in $X$.
The game $G$ has no other PNE.
}
\begin{proof}
It is clear that if one of the conditions (1)--(2) is satisfied then the game 
admits a PNE of the form described in the statement of the theorem.
For the converse direction, fix a tie-breaking rule $R\in\{R^C, R^V\}$ and a utility profile $\vecu$,
and suppose that a ballot vector $\vecb$ is a PNE of $(\calT, R,\vecu)$ with $|W(\vecb)|\ge 2$.
We will argue that $\vecu$ satisfies one of the conditions (1)--(2).

If $|W(\vecb)|=n$, each candidate in $W(\vecb)$ receives exactly one vote.
As argued in the proof of Theorem~\ref{thm:char-rand}, this means that
each voter votes for her favorite candidate, and prefers the uniform lottery
over $A=\{a_i\mid i\in N\}$ to her second most preferred candidate in $A$
being the unique winner, i.e., condition (1) holds.

Now, suppose that $|W(\vecb)|<n$. We claim that $b_i\in W(\vecb)$
for all $i\in N$. Indeed, suppose that $b_i\not\in W(\vecb)$ for some $i\in N$. 
Let $c_j$ be voter $i$'s most preferred candidate in $W(\vecb)$.
If $i$ changes her vote to $c_j$, $c_j$ becomes the unique winner, 
whereas when she votes $b_i$, the outcome is a lottery
over $W(\vecb)$ where candidates other than $c_j$ have a positive chance of winning.
Thus, $i$ can profitably deviate, a contradiction. Thus, there exists
a $k\ge 2$ such that each candidate in $W(\vecb)$ receives $n/k$ votes.
An argument similar to the one in the proof of Theorem~\ref{thm:char-rand}
shows that condition (2) must be satisfied.
\end{proof}

\smallskip

\noindent{\bf Theorem \ref{thm:truth-rand-single1}.\ }
{\em
Let $\vecu=(u_1, \dots, u_n)$ be a utility profile over $C$,
let $R\in\{R^C, R^V\}$, and suppose that $W(\veca)=\{c_j\}$ for some $c_j\in C$.
Then $\veca$ is a PNE of the game $G = (\calT, R, \vecu)$
if and only if for every $i\in N$ and every $c_k\in H(\veca)\setminus\{a_i\}$
it holds that $c_j\succ_i c_k$.
}
\begin{proof}
Consider the ballot vector $\veca$ and a voter $i\in N$.
Clearly, if $a_i=c_j$, voter $i$ cannot improve her utility by deviating.
Otherwise, the only way $i$ can change the election outcome is by changing
her vote to some $c_k\in H(\veca)\setminus\{a_i\}$, in which case the outcome is a lottery
over $\{c_j,c_k\}$ where both of these candidates has a positive chance of winning. 
The condition of the theorem says that no voter wants to change the 
election outcome in this way.
\end{proof}

\noindent{\bf Theorem \ref{thm:truth-rand-single2}.\ }
{\em
Let $\vecu=(u_1, \dots, u_n)$ be a utility profile over $C$,
let $R\in\{R^C, R^V\}$, and consider a ballot vector $\vecb$
with $W(\vecb)=\{c_j\}$ for some $c_j\in C$ and $b_r\neq a_r$ for some $r\in N$.
Then $\vecb$ is a PNE of the game $G = (\calT, R, \vecu)$
if and only if all of the following conditions hold:
\begin{itemize}
\item[(1)]
$b_i\in\{a_i, c_j\}$ for all $i\in N$;
\item[(2)]
$H(\vecb)\neq\emptyset$;
\item[(3)]
$c_j\succ_i c_k$ for all $i\in N$ and all $c_k\in H(\vecb)\setminus\{b_i\}$;
\item[(4)]
for every candidate $c_\ell\in H'(\vecb)$ 
and each voter $i\in N$ with $b_i=c_j$,
$i$ prefers $c_j$ to the lottery where a candidate
is chosen from $H(\vecb)\cup\{c_j, c_\ell\}$ according to $R$. 
\end{itemize}
}
\begin{proof}
Suppose that a ballot profile $\vecb$ satisfies conditions (1)--(4) 
of the theorem, and consider a voter $i\in N$.
If $b_i=a_i=c_j$, the current outcome is optimal for $i$.
If $b_i=a_i\neq c_j$, the only way that voter $i$ can change the election
outcome is by voting for a candidate $c_k\in H(\vecb)\setminus\{a_i\}$, 
in which case the winner will be chosen from $\{c_j,c_k\}$
according to $R$. By condition (3), voter $i$ does not benefit from this change.
By Proposition~\ref{prop:truth-basic}, the only remaining possibility is that $b_i=c_j\neq a_i$.
Then $i$ can change the election outcome by (a) voting for a candidate 
$c_k\in H(\vecb)$; (b) voting for a candidate $c_\ell\in H'(\vecb)$; or
(c) voting for a candidate in $C\setminus (H(\vecb)\cup H'(\vecb)\cup\{c_j\})$.
In case (a) $c_k$ becomes the unique winner, so by condition (3) this change
is not profitable to $i$. In case (b) the outcome is a tie among the candidates
in $H(\vecb)\cup\{c_j, c_\ell\}$, so by condition (4) voter $i$ cannot profit from
this change. Finally, in case (c) the outcome is a tie among the candidates
in $H(\vecb)\cup\{c_j\}$, and by condition (3), $i$ prefers the current outcome 
to this one. Thus, a ballot vector satisfying conditions (1)--(4) is indeed a PNE.

Conversely, suppose that $\vecb$ is a PNE of $(\calT, R, \vecu)$ for some $R\in\{R^C, R^V\}$
and some utility profile $\vecu$, where $b_r\neq a_r$ for some $r\in N$.
It follows from Proposition~\ref{prop:truth-basic}
that $\vecb$ satisfies condition (1). If condition (2) is violated, voter $r$
can increase her utility by $\eps$, by changing her vote to $a_r$, as $c_j$
would remain the unique election winner in this case. If condition (3) is violated
for some $i\in N$ and some $c_k\in H(\vecb)$, voter $i$ can profitably deviate
by changing her vote to $c_k$; if $b_i=c_j$, $c_k$ would then become 
the unique election winner, and if $b_i\neq c_j$, the outcome will
be a tie between $c_j$ and $c_k$, so under $R$ each of them 
will win with positive probability. Similarly, if condition (4) is violated
for some $i\in N$ and some $c_\ell\in H'(\vecb)$, voter $i$ can profitably deviate
by changing her vote to $c_\ell$, so that the outcome becomes a tie
among $H(\vecb)\cup\{c_j, c_\ell\}$.
This concludes the proof.
\end{proof}

\noindent 
{\bf Corollary \ref{cor:truth-rand-hard}.\ }
{\em
For $R\in\{R^C, R^V\}$, $(\calT, R)$-{\sc SingleNE},
$(\calT, R)$-{\sc TieNE}, and $(\calT, R)$-{\sc ExistNE}
are {\em NP}-complete.
}
\begin{proof}
Let $R\in\{R^C, R^V\}$. For $(\calT, R)$-{\sc TieNE}, as mentioned above, 
our claim follows from Theorem \ref{thm:char-rand} 
and its implications, as discussed in the section on lazy voters.

For $(\calT, R)$-{\sc SingleNE} with $R\in\{R^C, R^V\}$, we can use the same reduction from {\sc MSI} as in 
Theorem \ref{thm:truth-lex-hard}. The only change is the analysis in the last part of the proof of Theorem 
\ref{thm:truth-lex-hard}, due to the different tie-breaking rule. 
In particular, suppose again that a PNE $\vecb$ exists, where $w_2$ is the winner. 
By the analysis of Theorem \ref{thm:truth-lex-hard}, 
the set of candidates $\calE'=\{e\in\calE\mid\score(e, \vecb)=s-2\}$ contains at least $q$ elements.
Consider a candidate $e\in\calE'$. Suppose that among the voters 
from Block 1 who deviated to $w_2$, there exists a voter $v_{i_j}$ who prefers $e$ to $w_2$. 
Her utility in $\vecb$ is $1/2$. Suppose that she deviates to $e$ instead. In this case 
the score of $w_1$, $w_2$, and $e$ becomes $s-1$. Therefore, the new winning set is $\{w_1,w_2,e\}$. 
Given that $u_{i_j}(w_1) = 1/4$ and $u_{i_j}(e)> 3/4$, the utility of voter $v_{i_j}$ becomes more than 
$1/2$, contradicting the fact that $\vecb$ is a PNE.
Thus, for any voter $v_{i_j}$ in Block 1 who deviated to $w_2$, and for any candidate $e\in\calE'$ 
it holds that $v_{i_j}$ prefers $w_2$ to $e$, 
i.e., $e\in A_{i_j}$. 
Thus, $e\in A_{i_j}$ for each $j\in [k]$, where the sets $A_{i_1}, \dots, A_{i_k}$
are defined in the proof of Theorem 
\ref{thm:truth-lex-hard}. As this holds for every $e\in\calE'$ and $|\calE'|\ge q$, 
this means that we have a ``yes''-instance of {\sc MSI}. 
For the reverse direction, the arguments are very similar to the reverse direction in the proof of Theorem~\ref{thm:truth-lex-hard}.

Finally, regarding $(\calT, R^L)$-{\sc ExistNE}, a simple modification 
in the reduction of Theorem 
\ref{thm:truth-lex-hard} can yield the desired result; we omit the details from this version. 
\end{proof} 


\section{Price of Anarchy under Lexicographic Tie-breaking}
\label{app-poa}

We show that $\PoA = \Omega(n)$ both for lazy and for truth-biased voters 
under lexicographic tie-breaking.
In particular, we first establish that $\PoA = n-2$ for lazy voters. 
Then we show that $\PoA = 2n/3$ in the truth-biased model. 
Similar results can be proved for randomized tie-breaking, and we omit them 
from this version of the paper.

\begin{proposition}
\label{prop:lazy-poa}
For lexicographic tie-breaking and lazy voters, $\PoA = n-2$. 
\end{proposition}

\begin{proof}
We prove first that $\PoA\leq n-2$. To see this, note that by 
Theorem \ref{thm:lazy-lex-char}, the winner in any PNE must have 
a positive score in the truthful profile. 
Thus, in the worst-case scenario for the Price of Anarchy, 
the truthful winner of $\veca$ has score $n-1$, 
and there is a PNE where the winner is the candidate supported 
by the remaining voter. Thus $\PoA \leq n-2$.

To show the lower bound it suffices to exhibit an example. 
This is done in Example \ref{ex:lazy-poa}  below.  
\end{proof}

\begin{example}
\label{ex:lazy-poa}
{\em
Consider the lazy voters model and the profile of 
Figure~\ref{fig:lazy-poa}, with $n$ voters and $n$ candidates. 
It does not matter how we fill in the missing rankings in the figure. 
%
The truthful winner is $c_3$ with a score of $n-1$.
However, consider the profile $\vecb = (c_2, \bot, \bot, \ldots, \bot)$. 
The winner in $\vecb$ is $c_2$, and no voter can unilaterally change the outcome in her favor.
Indeed, if anyone votes for $c_1$, then $c_1$ is the new winner, 
but all voters prefer $c_2$ to $c_1$.
On the other hand, voting for any other candidate cannot 
change the outcome due to tie-breaking. Since the score of $c_2$ in $\veca$ is $1$, 
we have $\PoA \geq n-2$.

\begin{figure}[ht]
{
\begin{center}
\begin{tabular}{||ccccc||} \hline
$1$ & $2$ & $3$ & $\ldots$ & $n$  \\\hline
$c_2$ & $c_3$ & $c_3$ & $\ldots$ & $c_3$ \\
$\vdots$ & $c_2$ & $c_2$ & $\ldots$ & $c_2$ \\
$\vdots$ & $c_1$ & $c_1$ & $\ldots$ & $c_1$ \\
$\vdots$ & $\vdots$ & $\vdots$ & $\ldots$ & $\vdots$ \\\hline
\end{tabular}
\end{center}
}
\caption{$\PoA$ example for lazy voters.}
\label{fig:lazy-poa}
\end{figure}
}
\end{example}

\begin{proposition}
For lexicographic tie-breaking and truth-biased voters, 
$\PoA = 2n/3$. This holds even for single-peaked or single-crossing preference profiles.
\end{proposition}
\begin{proof}
As in Proposition \ref{prop:lazy-poa}, we first prove the upper bound.
Let $c_i$ be the winner in the truthful profile with a score of $s^*$. 
Let $\vecb\neq\veca$ be a non-truthful PNE and let $c_j$ be the winner in $\vecb$. 
Clearly, we have $\PoA \leq s^*$, since in the worst case $c_j$ 
has no supporters in $\veca$. Hence, it is enough to bound $s^*$.

By Lemma \ref{lem:threshold}, we know that there exists at least one threshold candidate with respect to $\vecb$.
We consider two cases:

\noindent {\bf Case 1:} 
$c_i\not\in T(\vecb)$. 
Then there is some $c_k\neq c_i$ such that $c_k\in T(\vecb)$
Let $s = \score(c_k, \veca)$. By Lemma \ref{lem:threshold} 
we know that $c_k$ receives $s$ points in $\vecb$ as well. 
Hence $c_j$ has a score of at most $s+1$ in $\vecb$. 
By Proposition \ref{prop:truth-basic} this means that there are at most $s+1$ non-truthful votes in $\vecb$. 
Hence the score of $c_i$ in $\vecb$ has to be at least $s^* - (s+1)$. 
Since $c_i$ is not a winner in $\vecb$, we have
$ s^* - (s+1) \leq \score(c_i, \vecb) \leq s+1$, 
and hence $s^* \leq 2s+2$.
Since the total score of  $c_i$ and $c_k$ in $\veca$ does not exceed $n$, 
we have $s+s^* \leq n$. But then, if $s^* > 2n/3$, 
this would imply that $s> n/3 -1$, i.e., $s\geq n/3$, and hence 
$s + s^* > n$, a contradiction. Thus we have $\PoA \leq s^* \leq 2n/3$.

\noindent {\bf Case 2:} $c_i\in T(\vecb)$. 
In this case the Price of Anarchy is somewhat better. 
Let $s=\score(c_j, \vecb)$. Candidate $c_i$ must have the same set of votes 
in $\vecb$ as in $\veca$ by Lemma \ref{lem:threshold}. Hence we have 
$s + s^* \leq n$. But we must also have $s^* \leq s$, 
otherwise $c_j$ is not the winner. 
But then if $s^* > n/2$, we would also have $s > n/2$, a contradiction. 
Thus, in this case we have $\PoA \leq s^* \leq n/2$.

Hence in worst case, $\PoA \leq 2n/3$. Finally, to show that the worst case is attained, 
we exhibit a construction in Example \ref{ex:truth-biased-poa}.
\end{proof}

\begin{example}
\label{ex:truth-biased-poa}
{\em
In Figure \ref{fig:truth-biased-PoA}, we show a preference profile for $n$ voters, where $n$ is divisible by $3$.
Block 1 consists of $n/3$ voters, Block 2 consists of $n/3+1$ voters, and Block 3 has $n/3 - 1$ voters.
In the figure, it does not matter how we fill in the missing rankings, 
but note that we can fill them in a way that makes the preference profile
single-peaked or single-crossing.

Suppose the tie-breaking rule is $c_1 > c_2 > c_3$. 
Under truthful voting, $c_3$ is the winner with a score of $2n/3$.
We claim now that the profile $\vecb$, in which all voters of Block 2 vote for $c_2$ is a PNE. 
To see this, note that $c_2$ is indeed the winner in $\vecb$ with a score of $n/3+1$. 
Candidate $c_1$ would only need one additional vote to become the winner, but 
there is no incentive for any voter from Block 2 or 3 to vote for $c_1$, 
since all of them prefer $c_2$ to $c_1$.
Also, no voter from Block 2 can change the outcome 
in favor of $c_3$ by a unilateral deviation, due to the tie-breaking rule.
If a voter from Block 2 switches to her truthful vote, 
then the new winner is $c_1$, since there is a tie with all candidates.
Hence $\vecb$ is a PNE, and the score of $c_2$ 
in the truthful profile is $0$. This means that in this example we have $\PoA\ge 2n/3$. 

\begin{figure}[ht]
{
\begin{tabular}{|*{4}{l}|*{4}{l}|*{4}{l}|}\hline
\multicolumn{4}{|c|}{Block 1}&\multicolumn{4}{|c|}{Block 2}&\multicolumn{4}{|c|}{Block 3}\\\hline\hline
$c_1$&$c_1$&$...$&$c_1$
&$c_3$&$c_3$&$...$&$c_3$
&$c_3$&$c_3$&$...$&$c_3$
\\
$\vdots$&$\vdots$&$...$&$\vdots$
&$\vdots$&$\vdots$&$...$&$\vdots$
&$\vdots$&$\vdots$&$...$&$\vdots$
\\
\multicolumn{4}{|c|}{arbitrary}
&\multicolumn{4}{|c|}{arbitrary}
&\multicolumn{4}{|c|}{arbitrary}
\\
$\vdots$&$\vdots$&$...$&$\vdots$
&$\vdots$&$\vdots$&$...$&$\vdots$
&$\vdots$&$\vdots$&$...$&$\vdots$
\\
$c_2$&$c_2$&$...$&$c_2$
&$c_2$&$c_2$&$...$&$c_2$
&$c_2$&$c_2$&$...$&$c_2$
\\
$c_3$&$c_3$&$...$&$c_3$
&$c_1$&$c_1$&$...$&$c_1$
&$c_1$&$c_1$&$...$&$c_1$
\\
\hline
\end{tabular}}
\caption{$\PoA$ example for truth-biased voters}
\label{fig:truth-biased-PoA}
\end{figure}
}
\end{example}

\section{Principled Voters}\label{app:principled}
We will now present our results for the setting with principled voters.
We omit the proofs of all results in this section, as they follow directly
from the analysis presented earlier in the paper.

In what follows, we consider elections with a set of strategic (i.e., lazy or truth-biased) 
voters $N=\{1, \dots, n\}$
and a set of principled voters $P=\{n+1,\dots, n+s\}$; we assume that either all voters in $N$ are lazy
or all of them are truth-biased. We extend our notation to such games as follows:
given a setting $\calS\in\{\calL, \calT\}$, a tie-breaking rule $R\in\{R^L, R^V, R^C\}$,
$n$ strategic voters with utilities $\vecu=(u_1, \dots, u_n)$, and $s$ principled
voters, whose votes are given by the ballot vector $\veca^P=(a_{n+1}, \dots, a_{n+s})$,
we denote the resulting game by $(\calS, R, \vecu, \veca^P)$. The principled voters
are not considered to be among the players; thus, the set of players in the modified game is still $N$.
As before, we use $\vecb$ to denote a ballot vector of the strategic voters;
$\vecb+\veca^P$ denotes a ballot vector that combines the votes of
the strategic and principled voters.
The computational problems {\sc ExistNE}, {\sc TieNE}, and {\sc SingleNE}
extend naturally to this setting;
we denote the respective variants of these problems by
{\sc ExistNE$^P$}, {\sc TieNE$^P$}, and {\sc SingleNE$^P$}, respectively.

\subsection{Principled + Lazy Voters, Lexicographic Tie-breaking}\label{sec:principled-lazy-lex}
\noindent
We have argued that in elections where all voters are lazy and the tie-breaking
rule is $R^L$, there is at most one voter who does not abstain and all PNE have the same winner. 
However, in the presence of principled voters, this is no longer true;
indeed, there are elections where {\em every} candidate can win in a PNE.

\begin{example}\label{ex:principled-lazy-lex}
{\em
Consider an election over a candidate set $C=\{c_1, \dots, c_m\}$, $m>1$, where
there are two principled voters who both vote for $c_m$, and two lazy voters
who both rank $c_m$ last. Then the ballot vector where both lazy voters abstain is 
a PNE (with winner $c_m$). Moreover, for every $j\in[m-1]$
the ballot vector where both lazy voters vote for $c_j$ is a PNE as well
(with winner $c_j$). 
}
\end{example}

Nevertheless, given an election with principled and lazy voters, we can characterize
the set of candidates who can win in a PNE of the respective game.

\begin{proposition}
Let $\vecu$ be the lazy voters' utility profile over $C$ and let $\veca^P$ be 
the principled voters' ballot vector. 
Let $j=\min\{k\mid c_k\in W(\veca^P)\}$, and let 
$H^+(\veca^P)=\{c_k\in H(\veca^P)\mid k<j\}$. 
Then  the game $G=(\calL, R^L, \vecu,\veca^P)$ has the following properties.
\begin{itemize}
\item[(1)]
If $\vecb$ is a PNE of $G$ then there is at most one candidate 
$c\in C$ such that $b_i=c$ for some $i\in N$; further, if
$b_i=c$ for some $c\in C$, $i\in N$, then $c$ is the winner in $\vecb+\veca^P$.
\item[(2)]
$G$ has a PNE where $c_j$ wins if and only if $(\bot, \dots,\bot)$ is a PNE of $G$. 
\item[(3)]
If $k>j$ then $G$ has a PNE where $c_k$ wins if and only if 
there are at least $M(\veca^P)+1-\score(c_k, \veca^P)$ lazy voters who prefer $c_k$ to all candidates in
$(W(\veca^P)\cup H^+(\veca^P))\setminus\{c_k\}$.
\item[(4)]
If $k<j$ then $G$ has a PNE where $c_k$ wins if and only if
there are at least $M(\veca^P)-\score(c_k, \veca^P)$ lazy voters who prefer $c_k$ to all candidates in
$(W(\veca^P)\cup H^+(\veca^P))\setminus\{c_k\}$.
\end{itemize} 
\end{proposition}

\begin{corollary}
The problems $(\calL, R^L)$-{\sc ExistNE$^P$}, $(\calL, R^L)$-{\sc TieNE$^P$}, and $(\calL, R^L)$-{\sc SingleNE$^P$}
are in {\em P}.
\end{corollary}

\subsection{Principled + Lazy Voters, Randomized Tie-breaking}\label{sec:principled-lazy-rand}
\noindent
We will now consider the effect of the presence of principled voters 
on lazy voters under randomized tie-breaking. We show that
single-winner PNE in this setting may have a more complicated structure
than single-winner PNE in the absence of principled voters.
On the other hand, PNE where several candidates are tied
for winning are very similar to those that arise when no principled
voters are present. We first consider the random candidate tie-breaking rule.

\begin{proposition}\label{prop:principled-lazy-rand-unique}
Let $\vecu=(u_1, \dots, u_n)$ be the lazy voters' utility profile over $C$, $|C|=m$,
and let $\veca^P=(a_{n+1}, \dots, a_{n+s})$ be the principled voters' ballot profile. 
The game $G = (\calL, R^C, \vecu, \veca^P)$ admits a PNE $\vecb$ with $W(\vecb+\veca^P)=\{c_j\}$
for some $c_j\in C$ 
if and only if one of the following conditions holds: 
\begin{itemize}
\item[(1)]
$W(\veca^P)=\{c_j\}$, $H(\veca^P)=\emptyset$;
\item[(2)]
$|V_j|\ge M(\veca^P)+1-\score(c_j, \veca^P)$, where $V_j$ is the set that consists
of all voters $i\in N$ such that (a) $u_i(c_j)>u_i(c_k)$ for all $c_k\in W(\veca^P)$
and (b) for each $c_\ell\in H(\veca^P)$ it holds that
$$
u_i(c_j)\ge \frac{1}{|W(\veca^P)+1|}\sum_{c\in W(\veca^P)\cup\{c_\ell\}}u_i(c).
$$ 
\end{itemize}
Moreover, if condition (1) holds then $G$ has a PNE where all lazy voters abstain, 
and if condition (2) holds then $G$ has a PNE where exactly
$M(\veca^P)+1-\score(c_j, \veca^P)$ lazy voters vote for $c_j$, 
while the remaining lazy voters abstain. The game $G$ has no other PNE
with winning set $\{c_j\}$.
\end{proposition}

\begin{corollary}\label{cor:principled-lazy-singleNE}
The problem $(\calL, R^C)$-{\sc SingleNE$^P$} is in {\em P}.
\end{corollary}

\begin{proposition}\label{prop:principled-lazy-rand-tie}
Let $\vecu=(u_1, \dots, u_n)$ be the lazy voters' utility profile over a candidate set $C$, $|C|=m$,
and let $\veca^P=(a_{n+1}, \dots, a_{n+s})$ be the principled voters' ballot profile.
Then the game $G = (\calL, R^C, \vecu, \veca^P)$ admits a PNE $\vecb$
with $|W(\vecb+\veca^P)|>1$ if and only if one of the following conditions holds:
\begin{itemize}
\item[(1)]
each candidate is ranked first by at most one voter in $N\cup P$
and $\frac{1}{n+s}\sum_{i\in N\cup P}u_\ell(a_i)\ge 
\max_{i\in (N\cup P)\setminus\{\ell\}}u_\ell(a_i)$ for all $\ell\in N$.  
\item[(2)]
there exists a set of candidates
$X = \{c_{\ell_1}, \dots, c_{\ell_k}\}$ with $k\ge 2$, 
a positive integer $n'\le n$ with $n'/k\ge 2$ such that
for each $c\not\in X$ we have $\score(c, \veca^P)<n'/k$,  
and a partition of the lazy voters into $k$ groups
$N_1, \dots, N_k$ (some of which may be empty) such that 
\begin{itemize}
\item[(a)] 
for each $j\in [k]$ we have $|N_j|+\score(c_{\ell_j}, \veca^P)= n'/k$;
\item[(b)]
for each $j\in[k]$ and each $i\in N_j$ we have $c_{\ell_j}\succ_i c$
for all $c\in X\setminus\{c_{\ell_j}\}$;
\item[(c)]
for each $j\in[k]$ and each $i\in N_j$ we have 
$\frac{1}{k}\sum_{c\in X}u_i(c)\ge \max_{c\in X\setminus\{c_{\ell_j}\}}u_i(c)$;
\item[(d)]
for each $j\in[k]$, each $i\in N_j$, and each $c'\in C\setminus X$ with $\score(c', \veca^P)=n'/k-1$
we have
$\frac{1}{k}\sum_{c\in X}u_i(c)\ge \frac{1}{k}\sum_{c\in (X\cup\{c'\})\setminus\{c_{\ell_j}\}}u_i(c)$.
\end{itemize}
\end{itemize}
Moreover, if condition (1) holds then $G$ has a PNE where 
each lazy voter votes for her top candidate, and if condition
(2) holds, then $G$ has a PNE where each lazy voter votes for her top candidate in $X$.
The game $G$ has no other PNE with two or more winners.
\end{proposition}

The following corollary is a direct consequence of Corollary~\ref{cor:lazy-rand-hard}
and the fact that the model with no principled voters is a special case 
of the model  with principled voters.
\begin{corollary}                          
The problems
$(\calL, R^C)$-{\sc TieNE$^P$} and $(\calL, R^C)$-{\sc ExistNE$^P$}
are {\em NP}-complete.
\end{corollary}

The reader may have noticed that 
Proposition~\ref{prop:principled-lazy-rand-unique}, 
Corollary~\ref{cor:principled-lazy-singleNE}, and Proposition~\ref{prop:principled-lazy-rand-tie} 
are stated for $R^C$, but not for $R^V$.
The reason for this is that in the presence of principled voters 
the tie-breaking rules $R^C$ and $R^V$
are no longer equivalent.

\begin{example}\label{ex:rc-neq-rv}
{\em
Consider an election over the candidate set $C=\{c_1, c_2, c_3\}$, 
where there are two lazy voters whose utility function 
is given by $u(c_1)=20$, $u(c_2)=4$, $u(c_3)=1$, two
lazy voters whose utility function 
is given by $u'(c_1)=20$, $u'(c_2)=4$, $u'(c_3)=1$, 
and one principled voter who ranks the candidates as $c_3\succ c_1\succ c_2$.
It is easy to see that both for $R^C$ and for $R^V$
the resulting game has a PNE where two lazy voters
vote for $c_1$, two lazy voters vote for $c_2$, and the principled
voter votes for $c_3$. Under $R^C$ candidates $c_1$ and $c_2$
are equally likely to win in this PNE. However, under $R^V$
candidate $c_1$ wins with probability $3/5$ and candidate $c_2$
wins with probability $2/5$. 
}
\end{example}

Nevertheless, all results in this section can be extended to random voter tie-breaking, 
by replacing the uniform lotteries over the winning sets 
in Propositions~\ref{prop:principled-lazy-rand-unique} and~\ref{prop:principled-lazy-rand-tie}
by lotteries that correspond to choosing an element of the winning set
according to the preferences of a random voter (who may be principled or lazy).
While the inequalities that one needs to verify become more cumbersome,
the complexity of the respective computational problems remains the same.
In particular, we obtain the following corollary.

\begin{corollary}
The problems
$(\calL, R^V)$-{\sc TieNE$^P$} and $(\calL, R^V)$-{\sc ExistNE$^P$}
are {\em NP}-complete, whereas
$(\calL, R^V)$-{\sc SingleNE$^P$} is in~{\em P}.
\end{corollary}

\subsection{Principled + Truth-biased Voters}
Principled and truth-biased voters are quite similar in their behavior;
therefore, adding principled voters to the setting of Section~\ref{sec:truth}
results in fewer changes than adding them to the setting
of Section~\ref{sec:lazy}. 

To illustrate this point, we will now show how to extend
Proposition~\ref{prop:truth-lex} to settings where principled voters may be present.

\begin{proposition}\label{prop:principled-truth-lex}
Let $\vecu$ be the utility profile of truth-biased voters, 
let $\veca$ be their truthful ballot vector, 
and let $\veca^P$ be the ballot vector of principled voters.
Let $j=\min\{r\mid c_r\in  W(\veca+\veca^P)\}$.
Then $\veca$ is a PNE of $(\calT, R^L, \vecu, \veca^P)$
if and only if neither of the following conditions holds:
\begin{itemize}
\item[(1)]
$|W(\veca+\veca^P)|>1$,
and there exists a candidate $c_k\in W(\veca+\veca^P)$
and a voter $i\in N$ such that $a_i\neq c_k$ and $c_k\succ_i c_j$.
\item[(2)]
$H(\veca+\veca^P)\neq\emptyset$, and there exists a candidate $c_k\in H(\veca+\veca^P)$
and a voter $i\in N$ such that $a_i\neq c_k$, $c_k\succ_i c_j$, and $k<j$.
\end{itemize}
\end{proposition}

All other claims in Section~\ref{sec:truth} can be modified in a similar way:
essentially, we replace $W(\vecb)$, $H(\vecb)$ and $H'(\vecb)$
with $W(\vecb+\veca^P)$, $H(\vecb+\veca^P)$ and $H'(\vecb+\veca^P)$,
but when considering the voters' incentives to change their votes, 
we limit our attention to truth-biased voters. Of course, 
we have to take into account the number of votes cast by principled
voters in favor of each candidate in the winning set,
and distinguish between $R^V$ and $R^C$ (as we did above).
Finally, it is immediate that all hardness results 
established in Section~\ref{sec:truth} remain true in the presence
of principled voters.

\end{document}